\documentclass[11pt]{article}

\usepackage[a4paper, margin=1in]{geometry}
\usepackage{amsmath, amssymb, amsfonts, amsthm}
\usepackage{mathtools}
\usepackage{booktabs}
\usepackage{array}
\usepackage{graphicx}
\usepackage{xcolor}
\usepackage[hidelinks]{hyperref}
\usepackage{caption}
\usepackage{enumitem}
\usepackage{microtype}
\usepackage{authblk}
\usepackage{cite}
\usepackage{lineno}

\definecolor{forestgreen}{rgb}{0.33,0.61,0.34}

\newtheorem{theorem}{Theorem}[section]
\newtheorem{proposition}[theorem]{Proposition}
\newtheorem{lemma}[theorem]{Lemma}
\newtheorem{corollary}[theorem]{Corollary}
\theoremstyle{definition}
\newtheorem{definitionbody}[theorem]{Definition}
\newtheorem{remarkbody}[theorem]{Remark}
\newtheorem{assumptionbody}[theorem]{Assumption}
\newcommand{\envend}{\nobreak\hfill$\lozenge$}
\newenvironment{definition}{\begin{definitionbody}}{\envend\end{definitionbody}}
\newenvironment{remark}{\begin{remarkbody}}{\envend\end{remarkbody}}
\newenvironment{assumption}{\begin{assumptionbody}}{\envend\end{assumptionbody}}

\newcommand{\Var}{\operatorname{Var}}
\newcommand{\Tr}{\operatorname{Tr}}
\newcommand{\E}{\mathbb{E}}
\newcommand{\sd}{\operatorname{sd}}
\newcommand{\bx}{\boldsymbol{x}}
\newcommand{\bp}{\boldsymbol{p}}
\newcommand{\br}{\boldsymbol{r}}
\newcommand{\bz}{\boldsymbol{z}}
\newcommand{\bv}{\boldsymbol{v}}
\newcommand{\bw}{\boldsymbol{w}}
\newcommand{\bu}{\boldsymbol{u}}
\newcommand{\by}{\boldsymbol{y}}
\newcommand{\bb}{\boldsymbol{b}}
\newcommand{\bde}{\boldsymbol{\delta}}
\newcommand{\bone}{\boldsymbol{1}}
\newcommand{\bW}{\boldsymbol{W}}

\title{A theory of spatial early warning signals for tipping points\\ on complex networks}

\author[1,2,3,*]{Naoki Masuda}
\affil[1]{Gilbert S.\ Omenn Department of Computational Medicine and Bioinformatics, University of Michigan, Ann Arbor, MI 48109, USA}
\affil[2]{Department of Mathematics, University of Michigan, Ann Arbor, MI 48109, USA}
\affil[3]{Center for Computational Social Science, Kobe University, Kobe, Hyogo 657-8501, Japan}
\affil[*]{E-mail: \texttt{naokimas@umich.edu}}
\date{}

\begin{document}
\maketitle


\begin{abstract}

Spatial early warning signals (EWSs) seek evidence of an approaching tipping point from a single snapshot of many interacting elements. Existing theory largely assumes spatial homogeneity, whereas networks introduce systematic differences among nodes that may obscure fluctuation-based warning signals. We develop a mathematical framework for spatial EWSs in stochastic dynamical systems on networks. We find that the expected spatial variance, a popular spatial EWS, decomposes exactly into a structural contribution from heterogeneity in the equilibrium state and a fluctuation contribution determined by the stationary covariance. Near a simple steady-state bifurcation, the potentially divergent covariance concentrates along the critical eigendirection: the left eigenvector determines how strongly noise excites the critical fluctuation, while the right eigenvector determines its spatial pattern. Consequently, the spatial variance has a divergent fluctuation contribution when the limiting critical eigendirection is noise-excited and spatially nonuniform after centering. In contrast, the spatial coefficient of variation generally saturates, while skewness, kurtosis, and Moran's $I$ approach network-dependent limits without a universal warning direction. We also derive results for homogeneous networks, node-wise baseline subtraction as preprocessing, and Hopf bifurcations, for which the limiting distributions are qualitatively different. These results clarify when spatial EWSs provide reliable warnings and why their performance depends on network structure, noise, and preprocessing.

\end{abstract}


\section{Introduction}\label{sec:intro}

Abrupt transitions between qualitatively distinct regimes recur across real systems. Dryland vegetation can collapse to a bare, desertified state once aridity passes a threshold~\cite{dakos2011slowing}, and several components of the climate system are thought to possess thresholds beyond which change becomes self-sustaining~\cite{Lenton2008PNAS, wunderling2022recurrent}. Because the consequences of abrupt transitions are often severe, a substantial literature has grown around early warning signals (EWSs), that is, statistics computed from observational data that change detectably before the transition occurs~\cite{Scheffer2009Nature, dakos2012methods, dakos2015resilience}. The mechanism most commonly invoked in EWSs is critical slowing down: as a control parameter approaches a bifurcation, the real part of the leading eigenvalue of the linearized dynamics at the stable state approaches zero from below, recovery from perturbations becomes progressively slower, and the stationary fluctuations grow and become more strongly autocorrelated. Important caveats about false positives and negatives attach to this program~\cite{Boettiger2012ProcRSoB, kefi2013early, rietkerk2025ambiguity, Boerlijst2013PLoSOne}, but critical slowing down remains what most EWSs attempt to detect.

The classical EWSs based on critical slowing down, such as the temporal variance and the lag-one autocorrelation, are estimated from a time series while the environment drifts slowly. That requirement is demanding in practice. A reliable estimate of even one variance needs many samples~\cite{JCGM100_2008}, and those samples must be gathered while the control parameter is effectively constant, which is hard when observation is expensive, destructive, or slow~\cite{Dakos2010TheorEcol}. Rolling windows only partly relieve the difficulty, at the cost of smoothing the very trend one is trying to detect~\cite{dakos2012methods, gsell2016evaluating}. Model-fitting and machine-learning approaches, although they broaden the toolkit considerably, still generally presuppose reasonably long records~\cite{boettiger2012quantifying, hessler2022bayesian, kong2021machine, patel2023using, liu2024early}. Spatial EWSs offer an alternative by trading the time axis for the spatial one. Instead of watching one element repeatedly, one records a single snapshot across many interacting elements in the system in question, which we call nodes throughout, and computes a statistic across these nodes~\cite{Dakos2010TheorEcol, guttal2009spatial, kefi2014early, nijp2019spatial}. Thus, one snapshot per environmental condition suffices for computing and monitoring a spatial EWS. The sample variance, skewness, and kurtosis computed from the nodes' states are examples of spatial EWSs.

Most studies of spatial EWSs have considered spatially homogeneous domains, such as regular square lattices and their continuum limits~\cite{guttal2009spatial, Dakos2010TheorEcol, dakos2011slowing, kefi2014early, sankaran2018implications}. Likewise, analytical theory for EWSs in systems with more than a few variables has focused almost entirely on spatially homogeneous continua, particularly using stochastic partial differential equations~\cite{Gowda2015CNSNS, Kuehn2019EJAM, Bernuzzi2025EJAM, Bernuzzi2026PDEA, Clarke2026JPhysComplexity} (see~\cite{guttal2009spatial} for a mean-field approach and~\cite{tirabassi2024linear} for a networked ODE approach to homogeneous media). Empirical systems of interest, however, are rarely so regular. Ecological, epidemiological, genetic, infrastructural, and social systems in which tipping events occur are typically organized as heterogeneous networks~\cite{Newman2018book, LiuLiMa2022PhysRep}, whose nodes differ in degree and other properties and may therefore occupy different equilibrium states even far from a transition. Numerical simulations have shown that classical spatial EWSs behave erratically on such networks: which indicator performs best, and even whether it rises or falls as a bifurcation is approached, depends on the dynamical system, control parameter, direction of parameter variation, and network~\cite{maclaren2025applicability, robinson2025assessing}. We recently proposed referencing each node's state to its own baseline value measured far from the bifurcation before computing a spatial EWS; this procedure improves the performance of the spatial variance in particular~\cite{Bandara2026arxiv}.

These numerical findings raise more questions than they answer because the mechanisms governing spatial EWSs in heterogeneous networks remain poorly understood. What do these commonly used, heuristically defined statistics actually measure? Why do many perform well in homogeneous media but not on heterogeneous networks~\cite{maclaren2025applicability, robinson2025assessing}? Why should subtracting a node-wise baseline rescue the spatial variance while leaving skewness, kurtosis, and Moran's $I$ largely unimproved~\cite{Bandara2026arxiv}? This article develops a mathematical foundation for these and related questions. We linearize stochastic network dynamics about a stable equilibrium and analyze statistics computed across nodes from a single snapshot as the equilibrium loses stability. Our organizing observation is that, on a heterogeneous network, the expected spatial variance decomposes into a deterministic contribution determined by variation among equilibrium node states and a stochastic contribution determined by the stationary covariance. This distinction disappears on a homogeneous domain, where the equilibrium is uniform. We determine which of five standard spatial EWSs diverge near a loss of stability, quantify the uncertainty arising from a single snapshot, and establish what baseline referencing does and does not accomplish.

\section{Model and preliminaries}\label{sec:setup}

\subsection{Stochastic dynamics on networks}

We consider $N$ dynamical elements, which we call nodes, coupled on a connected network with adjacency matrix $A=(A_{ij})$ or, more generally, a different interaction matrix. When the interactions are directed, connectedness here means weak connectedness. Unless stated otherwise, we do not assume that the interaction matrix is symmetric or binary, so directed and weighted interactions are allowed. 
Let $x_i(t)\in\mathbb{R}$ denote the state of the $i$th node and $\bx(t)=(x_1(t),\ldots,x_N(t))^\top$, where ${}^{\top}$ represents the transposition. We consider the It\^o stochastic differential equation
\begin{equation}\label{eq:sde}
\mathrm{d}\bx(t) = F(\bx(t); \alpha)\,\mathrm{d}t + B\,\mathrm{d}\bW(t),
\end{equation}
where $\alpha$ is a scalar control parameter, $F:\mathbb{R}^N\to\mathbb{R}^N$ collects the self-dynamics and the network coupling, $\bW$ is an $m$-dimensional vector of independent standard Wiener processes, and $B \in \mathbb{R}^{N\times m}$ is a constant matrix. Only the diffusion matrix $BB^\top\in\mathbb{R}^{N\times N}$ enters the results below, so the number of independent noise sources, $m$, plays no role. We write $B=\mathrm{diag}(\sigma_1,\ldots,\sigma_N)$ when the noise is independent across nodes with node-dependent strength. We assume that $F$ is sufficiently smooth in $(\bx,\alpha)$ for the derivatives and local bifurcation expansions used below; $C^2$ regularity in a neighborhood of the relevant equilibrium branch is sufficient.

An example is the coupled double-well system~\cite{brummitt2015coupled, wunderling2022recurrent, MacLaren2023JRoySocInterface}
\begin{equation}\label{eq:doublewell}
\mathrm{d}x_i = \Big[-(x_i-r_1)(x_i-r_2)(x_i-r_3) + D\sum_{j=1}^N A_{ij}x_j + u\Big]\mathrm{d}t + \sigma\,\mathrm{d}W_i,
\end{equation}
with $r_1<r_2<r_3$, coupling strength $D$ and stress $u$; either $D$ or $u$ may play the role of $\alpha$.

Let $\bx^*(\alpha)$ denote a stable equilibrium of the noise-free dynamics, satisfying $F(\bx^*;\alpha)=\boldsymbol 0$. We define the stability matrix
\begin{equation}\label{eq:M}
M(\alpha) \;=\; -\left.\frac{\partial F}{\partial \bx}\right|_{\bx=\bx^*(\alpha)},
\end{equation}
that is, the sign-flipped Jacobian, so that $\bx^*$ is asymptotically stable when all eigenvalues of $M$ have positive real parts.
For the coupled double-well system~\eqref{eq:doublewell}, $M = \mathrm{diag}\big(p'(x_i^*)\big) - DA$ with $p(x)=(x-r_1)(x-r_2)(x-r_3)$, which is symmetric whenever $A$ is.

We study a one-parameter family of such systems and follow a branch of stable equilibria as the control parameter $\alpha$ varies. The situation of interest is the one in which this branch persists over a range of $\alpha$ and then loses stability at a critical value $\alpha=\alpha_{\mathrm c}$, after which the dynamics may move toward a qualitatively different state. Anticipating this event from data recorded while the system is still at $\alpha<\alpha_{\mathrm c}$, and while nothing dramatic has yet happened to the node states, is the entire purpose of an EWS. For a steady-state bifurcation, the loss of stability occurs when a real eigenvalue of the Jacobian reaches zero, equivalently when an eigenvalue of the stability matrix $M$ reaches zero. The approach of that eigenvalue to zero produces critical slowing down. The following assumption specifies the type of destabilization treated in the main part of this paper. It includes the standard saddle-node, transcritical, and pitchfork bifurcations, although transcritical and pitchfork bifurcations generally require additional conditions.

\begin{assumption}[Regular simple real critical eigenvalue]\label{ass:spec}
Along the stable equilibrium branch, let $\alpha \to \alpha_{\mathrm c}^-$ denote approach from below to the critical parameter value $\alpha_{\mathrm c}$ at which $\bx^*(\alpha)$ loses stability. We assume that the equilibrium branch has a finite limit, denoted by $\bx_{\mathrm c}^*$. For $\alpha<\alpha_{\mathrm c}$ sufficiently close to $\alpha_{\mathrm c}$,
the real matrix $M(\alpha)$ is diagonalizable over $\mathbb C$ and all of its eigenvalues have positive real parts.
Its eigenvalue of smallest real part, denoted by $\lambda_1(\alpha)$, is real and simple, with
\begin{equation}\label{eq:specgapass}
\lambda_1(\alpha)>0,\qquad \lim_{\alpha\to\alpha_{\mathrm c}^-}\lambda_1(\alpha)=0,\qquad \operatorname{Re}\lambda_k(\alpha)\ge\delta>0\quad(k=2,\ldots,N).
\end{equation}
Thus, only one real eigendirection becomes critical.
The right and left eigenvectors can be chosen as biorthogonal families,
\begin{equation}
M\bv^{(k)}=\lambda_k\bv^{(k)},\qquad \bw^{(k)\top}M=\lambda_k\bw^{(k)\top},\qquad \bw^{(k)\top}\bv^{(\ell)}=\delta_{k \ell},
\end{equation}
with $\|\bv^{(k)}\|_2=1$. Because $M(\alpha)$ is real and $\lambda_1(\alpha)$ is real and simple, its one-dimensional right and left critical eigenspaces contain real nonzero vectors~\cite{horn2013matrix}. We therefore choose $\bv\equiv\bv^{(1)}$ and $\bw\equiv\bw^{(1)}$ to be real, with the above normalization, and assume that the finite limits
\begin{equation}\label{eq:criticalvectors}
\bv_{\mathrm c}\equiv\lim_{\alpha\to\alpha_{\mathrm c}^-}\bv(\alpha),\qquad
\bw_{\mathrm c}\equiv\lim_{\alpha\to\alpha_{\mathrm c}^-}\bw(\alpha)
\end{equation}
exist.
\end{assumption}

\begin{remark}
Assumption~\ref{ass:spec} describes a loss of stability in which one real eigenvalue approaches zero. It therefore does not cover a Hopf bifurcation, where a complex-conjugate pair approaches the imaginary axis with nonzero imaginary parts. The Hopf case is treated separately in Section~\ref{sec:hopf}.
\end{remark}

\subsection{Linearization and the stationary covariance}

Writing $\bz(t)=\bx(t)-\bx^*(\alpha)$ and linearizing~\eqref{eq:sde} gives the multivariate OU process
\begin{equation}\label{eq:ou}
\mathrm{d}\bz(t) = -M\bz(t)\,\mathrm{d}t + B\,\mathrm{d}\bW(t),
\end{equation}
whose stationary distribution is Gaussian with mean $\boldsymbol0$ and covariance $C=(C_{ij})$ solving the Lyapunov equation~\cite{gardiner2009stochastic, gajic1995lyapunov}
\begin{equation}\label{eq:lyapunov}
MC + CM^\top = BB^\top.
\end{equation}
Under Assumption~\ref{ass:spec}, expanding the stationary covariance in the biorthogonal eigenvector basis gives~\cite{gardiner2009stochastic, oku2018covariance, Patterson2021AmNat}
\begin{equation}\label{eq:Cspectral}
C=\sum_{k=1}^{N}\sum_{\ell =1}^{N}
\frac{\beta_{k \ell}}{\lambda_k+\lambda_{\ell}}\,
\bv^{(k)}\big(\bv^{(\ell)}\big)^\top,
\end{equation}
where
\begin{equation}
\beta_{kl}\equiv
\bw^{(k)\top}BB^\top\bw^{(l)}.
\end{equation}
Indeed, substituting
$C=\sum_{k, \ell = 1}^N c_{k \ell}\bv^{(k)}\big(\bv^{(\ell)}\big)^\top$
into~\eqref{eq:lyapunov} yields
$c_{k \ell}(\lambda_k+\lambda_\ell)=\beta_{k \ell}$.

Equation~\eqref{eq:Cspectral} separates two distinct roles. The noise enters through the left eigenvectors $\bw^{(k)}$, whereas the spatial pattern of the fluctuations is carried by the right eigenvectors $\bv^{(k)}$. The left and right eigenvectors can be chosen to coincide in the symmetric case. If $M$ is symmetric and $B=\sigma I$, then
$C=\frac{\sigma^2}{2}M^{-1}$.

\begin{lemma}[Rank-one asymptotics of the covariance~\cite{neumaier2001estimation, ives2003estimating, oku2018covariance, Patterson2021AmNat}]\label{lem:rank1}
Under Assumption~\ref{ass:spec}, as
$\alpha\to\alpha_{\mathrm c}^-$, equivalently as
$\lambda_1\to0^+$,
\begin{equation}\label{eq:rank1}
C=
\frac{\sigma_{\mathrm{eff}}^2}{2\lambda_1}
\,\bv\bv^\top
+O(1),
\end{equation}
where
\begin{equation}\label{eq:sigma_eff-def}
\sigma_{\mathrm{eff}}^2
\equiv
\bw^\top BB^\top\bw.
\end{equation}
The $O(1)$ term in~\eqref{eq:rank1} is bounded in operator norm. If
$B=\operatorname{diag}(\sigma_1,\ldots,\sigma_N)$, then
$\sigma_{\mathrm{eff}}^2
=
\sum_{i=1}^N \sigma_i^2 w_i^2$.
\end{lemma}

We prove Lemma~\ref{lem:rank1} in Appendix~\ref{app:rank1-proof}.

Because $B$ is fixed and $\bw\to\bw_{\mathrm c}$, we define the limiting effective noise intensity by
\begin{equation}\label{eq:sigma_eff_c}
\sigma_{\mathrm{eff,c}}^2
=
\bw_{\mathrm c}^\top BB^\top\bw_{\mathrm c}
=
\lim_{\alpha\to\alpha_{\mathrm c}^-}
\sigma_{\mathrm{eff}}^2.
\end{equation}
We say that the critical eigendirection is noise-excited in the limit when
$\sigma_{\mathrm{eff,c}}^2>0$.

\subsection{Spatial early warning signals}

We assume throughout that the observation protocol is the standard one for spatial EWSs~\cite{Dakos2010TheorEcol, maclaren2025applicability}: at each value of $\alpha$ the system is allowed to equilibrate under the noise, and a single sample of each node state is recorded. The observed snapshot is therefore
\begin{equation}\label{eq:snapshot}
x_i = x_i^*(\alpha) + z_i,
\end{equation}
where $\bz = (z_1, \ldots, z_N)$ has distribution $\mathcal{N}(\boldsymbol 0, C(\alpha))$,
and all expectations below are over the stationary law of $\bz$ at fixed $\alpha$.

For an arbitrary vector $\bu = (u_1, \ldots, u_N) \in\mathbb R^N$, we write $\bar u = N^{-1}\sum_i u_i$. We define the sample variance and the centered moments by
\begin{equation}
\Var(\bu) = \frac{1}{N-1}\sum_{i=1}^N (u_i-\bar u)^2
\end{equation}
and
\begin{equation}
\mu_k(\bu) = \frac{1}{N}\sum_{i=1}^N (u_i-\bar u)^k,
\end{equation}
respectively.

Given the observed snapshot $\bx$,
we analyze the following popular spatial EWSs~\cite{maclaren2025applicability, Bandara2026arxiv}:
\begin{align}
V &= \Var(\bx) && \text{spatial variance~\cite{guttal2009spatial, eby2017alternative, buelo2018modeling, ma2022spatiotemporal}},\label{eq:V}\\
\mathrm{CV} &= \sqrt{V}\,/\,|\bar x| && \text{spatial coefficient of variation~\cite{dai2013slower, litzow2008increased, rindi2018experimental}},\label{eq:CV}\\
g_1 &= \mu_3(\bx)\,/\,\mu_2(\bx)^{3/2} && \text{spatial skewness~\cite{guttal2009spatial, buelo2018modeling, ma2022spatiotemporal}},\label{eq:g1}\\
g_2 &= \mu_4(\bx)\,/\,\mu_2(\bx)^{2} && \text{spatial kurtosis~\cite{buelo2018modeling, ma2022spatiotemporal}},\label{eq:g2}\\
I_{\mathrm M} &= \frac{N}{W}\,\frac{\sum_{i,j} A_{ij}(x_i-\bar x)(x_j-\bar x)}{\sum_i (x_i-\bar x)^2} && \text{Moran's }I~\text{\cite{legendre1989spatial, Dakos2010TheorEcol}},\label{eq:moran}
\end{align}
with $W=\sum_{i,j = 1}^N A_{ij}>0$. We omit the spatial standard deviation because it is just equal to $\sqrt{V}$. Only $I_{\mathrm M}$ requires knowledge of the adjacency matrix. When $\bar x>0$, in particular when all node values are positive, definition~\eqref{eq:CV} agrees with the usual CV given by $\sqrt{V}/\bar x$; our CV is undefined when $\bar x=0$.

We illustrate the behavior of interest. Figure~\ref{fig:ews} shows the five spatial EWSs defined by Eqs.~\eqref{eq:V}--\eqref{eq:moran} for the coupled double-well dynamics, given by~\eqref{eq:doublewell}, on five networks, as the dynamical system is driven toward a tipping point either by gradually decreasing the coupling strength $D$ (Fig.~\ref{fig:ews}(a)--(e)) or by gradually decreasing the stress $u$ (Fig.~\ref{fig:ews}(f)--(j)). Each column of the figure uses the same network for the two control parameters. At each value of the control parameter, we equilibrate the stochastic dynamics and record one sample of each node's state, i.e., a single snapshot $\{x_1,\ldots,x_N\}$, from which we compute the five EWSs; we only use the values of the control parameter that precede the first tipping event. We describe the details of the numerical simulations and the networks in Appendix~\ref{app:sim}. Figure~\ref{fig:ews} indicates that the five EWSs disagree with each other across networks and across the two control parameters, and the different behavior is not systematically related to the number of nodes, $N$. Furthermore, even the same EWS tends to behave differently across networks and control parameters. The same EWS on the same network can even reverse its direction when we change the control parameter.
These inconsistencies across EWSs, networks, and control parameters were also reported in previous numerical work~\cite{maclaren2025applicability, robinson2025assessing, Bandara2026arxiv} and are what the theory in the following sections aims to explain, at least partially.

\begin{figure}[t]
\centering
\includegraphics[width=\textwidth]{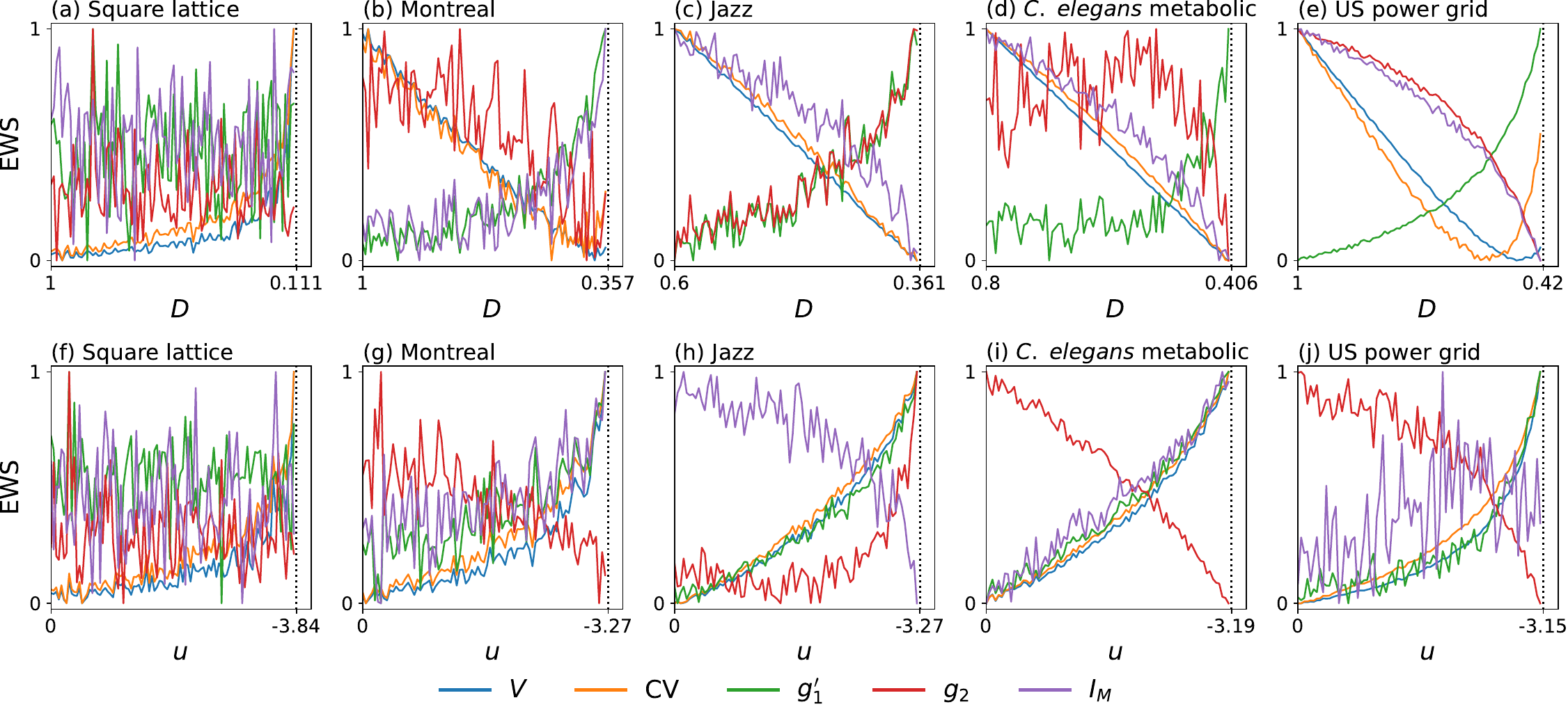}
\caption{Diverse behavior of the five spatial EWSs across networks and control parameters. We use the coupled double-well dynamics~\eqref{eq:doublewell}. Each panel shows the five spatial EWSs from the control-parameter value farthest from the tipping point (left) to the first tipping event (right; dotted line). Because both rows use a descending control parameter $D$ or $u$, we plot $g_1'=-g_1$, so that an increase in $g_1'$ signals proximity to the tipping point in the same manner as for the other four EWSs. We min--max normalize each curve to $[0,1]$ within each panel. In other words, we linearly scale the smallest and largest values that the EWS takes to $0$ and $1$, respectively, for visual purposes. (a)--(e) Descending coupling strength $D$. (f)--(j) Descending stress $u$. Each column uses the same network, i.e., the square lattice ($N=100$), Montreal ($N=29$), jazz player ($N=198$), \textit{C.\ elegans} metabolic ($N=453$), and US power grid ($N=4941$) networks. See Appendix~\ref{app:sim} for the networks.}
\label{fig:ews}
\end{figure}

\section{Structural and fluctuation contributions to the spatial variance}\label{sec:decomp}

\subsection{Exact decomposition}\label{sub:decomp}

Let $P = I - \tfrac1N\bone\bone^\top$ denote the centering projector, so that $V = \|P\bx\|^2/(N-1)$.
The following elementary observation separates the two mechanisms that drive a spatial statistic.

\begin{proposition}[Structural--fluctuation decomposition]\label{prop:decomp}
Under the observation model~\eqref{eq:snapshot}, we obtain
\begin{equation}\label{eq:decomp}
\E[V] \;=\; \underbrace{\Var\big(\bx^*(\alpha)\big)}_{\textstyle S(\alpha)} \;+\; \underbrace{\frac{1}{N-1}\Tr\big(PC(\alpha)\big)}_{\textstyle \Phi(\alpha)}.
\end{equation}
\end{proposition}

\begin{proof}
By using $V=\|P\bx\|^2/(N-1)$ and $P\bx=P\bx^*+P\bz$, and taking expectations, we obtain
\begin{align}
\E[V]
&=\frac{1}{N-1}\E\left[\|P\bx^*+P\bz\|^2\right] \notag \\
&=\frac{\|P\bx^*\|^2}{N-1}
+\frac{2(P\bx^*)^\top P\E[\bz]}{N-1}
+\frac{\E[\bz^\top P\bz]}{N-1}.
\end{align}
The first term is $\Var(\bx^*)$ by definition. The middle term is zero because $\E[\bz]=\boldsymbol0$. We also obtain $\E[\bz^\top P\bz]=\Tr(P\E[\bz\bz^\top])=\Tr(PC)$, concluding the proof.
\end{proof}

We call $S(\alpha)$ the structural term and $\Phi(\alpha)$ the fluctuation term. Explicitly,
\begin{align}
S(\alpha) &= \frac{1}{N-1}\sum_{i=1}^N\big(x_i^*(\alpha)-\overline{x^*}(\alpha)\big)^2, \\
\Phi(\alpha) &= \frac{1}{N-1}\left(\Tr C - \frac{1}{N}\bone^\top C\bone\right).
\label{eq:Phi(alpha)-def}
\end{align}

Proposition~\ref{prop:decomp} is the structural reason why the theory of spatial EWSs on heterogeneous networks differs from the theory on homogeneous domains such as $\mathbb{Z}^2$ and $\mathbb{R}^2$. On the uniform equilibrium branch usually studied in a spatially homogeneous system, $x_i^*$ is independent of $i$, and hence $S(\alpha)=0$. The spatial variance then measures fluctuations and nothing else, which is the situation analyzed in, e.g.,~\cite{guttal2009spatial, dakos2011slowing, tirabassi2024linear, Clarke2026JPhysComplexity}. On a heterogeneous network, $S(\alpha)$ is generically nonzero and $O(1)$, because equilibrium node states differ systematically with degree and position even far from any bifurcation of the dynamical system. The fluctuation term $\Phi(\alpha)$ is the expected sample variance across nodes of the stochastic fluctuation vector $\bz$. In \eqref{eq:Phi(alpha)-def}, $\Tr C$ sums the fluctuation variances of the individual nodes, whereas $N^{-1}\bone^\top C\bone$ removes the part associated with fluctuations of the spatial mean; thus a fluctuation that shifts all nodes together does not contribute to $\Phi(\alpha)$.
The spatial variance, $V$, combines $S(\alpha)$ and $\Phi(\alpha)$.

\begin{remark}
The decomposition is exact and requires neither linearization nor proximity to a bifurcation; only $\E[\bz]=\boldsymbol 0$ is used. Linearization enters below solely through the evaluation of $C$.
\end{remark}

The decomposition reduces the problem to understanding the deterministic structural term $S(\alpha)$ and the fluctuation term $\Phi(\alpha)$ separately. We first examine $S(\alpha)$.

\subsection{The structural term remains bounded}\label{sub:Sbounded}

At a generic saddle-node bifurcation, the center-manifold normal form has a quadratic tangency. If $\varepsilon=\alpha_{\mathrm c}-\alpha$ measures the distance to the bifurcation on the stable side, the equilibrium displacement along the center direction is of order $\sqrt{\varepsilon}$, and the leading eigenvalue therefore satisfies $\lambda_1(\alpha)=\kappa\sqrt{\varepsilon}+O(\varepsilon)$ with $\kappa>0$ under our sign convention for $M$~\cite{strogatz2024nonlinear,kuznetsov2004elements, Bury2020JRSocInterface}.

\begin{proposition}[Behavior of the structural term]\label{prop:S}
Suppose the equilibrium branch terminates in a generic saddle-node bifurcation and write $\varepsilon=\alpha_{\mathrm c}-\alpha$. Let $\lambda_1(\alpha)=\kappa\sqrt{\varepsilon}+O(\varepsilon)$ with $\kappa>0$. Let $\bx_{\mathrm c}^*\equiv\lim_{\alpha\to\alpha_{\mathrm c}^-}\bx^*(\alpha)$, and let $\bv_{\mathrm c}$ and $\bw_{\mathrm c}$ be the corresponding limits of the right and left critical eigenvectors defined in~\eqref{eq:criticalvectors}. Then,
\begin{equation}\label{eq:Sbounded}
S(\alpha)\longrightarrow S(\alpha_{\mathrm c})=\Var(\bx_{\mathrm c}^*)<\infty,
\end{equation}
and
\begin{equation}\label{eq:Sasym}
S(\alpha)=S(\alpha_{\mathrm c})-2\kappa_S\sqrt{\varepsilon}+O(\varepsilon),
\end{equation}
where
\begin{equation}
\kappa_S=\frac{2}{\kappa}\frac{\big[\bw_{\mathrm c}^\top\partial_\alpha F(\bx_{\mathrm c}^*;\alpha_{\mathrm c})\big](P\bx_{\mathrm c}^*)^\top P\bv_{\mathrm c}}{N-1}.
\end{equation}
Therefore, $S(\alpha)$ is bounded, but its derivative generically diverges with a sign that need not be positive.
\end{proposition}

\begin{proof}
Continuity of the equilibrium branch gives boundedness. Differentiating the equilibrium identity with respect to $\alpha$ gives
\begin{equation}
\boldsymbol0
=\frac{\mathrm d}{\mathrm d\alpha}F(\bx^*(\alpha);\alpha)
=\left.\frac{\partial F}{\partial\bx}\right|_{\bx=\bx^*(\alpha)}\frac{\mathrm d\bx^*}{\mathrm d\alpha}+\partial_\alpha F
=-M\frac{\mathrm d\bx^*}{\mathrm d\alpha}+\partial_\alpha F.
\end{equation}
Therefore,
\begin{equation}
\frac{\mathrm d\bx^*}{\mathrm d\alpha}=M^{-1}\partial_\alpha F.
\end{equation}
Using the simple-eigenvalue expansion
$M^{-1}
=
\lambda_1^{-1}\bv\bw^\top+O(1)$,
we obtain
\begin{equation}
\frac{\mathrm d\bx^*}{\mathrm d\alpha}
=
\frac{1}{\lambda_1}
\bv\bw^\top\partial_\alpha F
+
O(1).
\label{eq:proof-smooth-0}
\end{equation}
For a generic saddle-node bifurcation, the equilibrium branch has a smooth local parameter $\xi$ such that $\xi=O(\sqrt{\varepsilon})$, where $\varepsilon=\alpha_{\mathrm c}-\alpha$. Because the zero eigenvalue at the bifurcation is simple, its left and right eigenvectors, as well as $\partial_\alpha F$ evaluated along the equilibrium branch, can be chosen to depend smoothly on $\xi$. Therefore,
\begin{equation}
\bv\bw^\top\partial_\alpha F
=
\bv_{\mathrm c}
\bw_{\mathrm c}^\top
\partial_\alpha F(\bx_{\mathrm c}^*;\alpha_{\mathrm c})
+
O(\sqrt{\varepsilon}).
\label{eq:proof-smooth-1}
\end{equation}
Since $\lambda_1=\kappa\sqrt{\varepsilon}+O(\varepsilon)$, division of the $O(\sqrt{\varepsilon})$ term on the right-hand side of \eqref{eq:proof-smooth-1} by $\lambda_1$ contributes only $O(1)$. Therefore, \eqref{eq:proof-smooth-0} leads to
\begin{equation}
\frac{\mathrm d\bx^*}{\mathrm d\alpha}
=
\frac{
\bw_{\mathrm c}^\top
\partial_\alpha F(\bx_{\mathrm c}^*;\alpha_{\mathrm c})
}{\lambda_1}
\bv_{\mathrm c}
+
O(1).
\label{eq:proof-smooth-2}
\end{equation}
Integrating \eqref{eq:proof-smooth-2} from $\alpha$ to $\alpha_{\mathrm c}$ and using
$\lambda_1=\kappa\sqrt{\varepsilon}+O(\varepsilon)$ gives
\begin{equation}
\bx^*(\alpha)
=
\bx_{\mathrm c}^*
-
\frac{
2\bw_{\mathrm c}^\top
\partial_\alpha F(\bx_{\mathrm c}^*;\alpha_{\mathrm c})
}{\kappa}
\sqrt{\varepsilon}\,\bv_{\mathrm c}
+
O(\varepsilon).
\label{eq:integrated-alpha}
\end{equation}
Substitution of \eqref{eq:integrated-alpha} into $S=\|P\bx^*\|^2/(N-1)$ gives~\eqref{eq:Sasym}.
\end{proof}

Two consequences organize the rest of the paper. First, because $S(\alpha)$ remains bounded, it does not by itself create an unbounded warning signal. Any divergence of the expected spatial variance must come from the fluctuation term $\Phi(\alpha)$. Second, the sign of $\kappa_S$, which determines the leading local trend of the structural term $S(\alpha)$ near the bifurcation, can be positive or negative, depending on the dynamical system, the network, the control parameter, and the direction in which the control parameter is moved. The structural contribution can therefore reinforce, mask, or even reverse the trend of $V$ originating from the fluctuation term $\Phi(\alpha)$ as a bifurcation is gradually approached.

\section{Convergence and divergence of spatial early warning signals}\label{sec:main}

In this section, we investigate the behavior of the five classical spatial EWSs in~\eqref{eq:V}--\eqref{eq:moran} as a bifurcation is approached, with particular attention to whether they converge or diverge.

\subsection{Critical-coordinate representation of a snapshot}

Lemma~\ref{lem:rank1} describes the ensemble fluctuations. To analyze statistics computed from an individual snapshot of the network, $\bx$, we first derive an asymptotic representation of an individual snapshot near the bifurcation. We recall that
$\bx(\alpha) = \bx^*(\alpha) + \bz$ and that $\sigma_{\mathrm{eff}}^2$ is defined in~\eqref{eq:sigma_eff-def}.

\begin{lemma}[Critical-coordinate decomposition]\label{lem:coordinate}
Assume that $\sigma_{\mathrm{eff,c}}^2>0$.
Define
\begin{equation}\label{eq:coordinate}
\gamma=\bw^\top\bx=m+\zeta,
\end{equation}
where
\begin{align}
m(\alpha) &= \bw^\top \bx^*(\alpha),\\
\zeta &= \bw^\top\bz.
\end{align}
Then, $\zeta$ has distribution $\mathcal N(0,s^2)$, where
$s^2(\alpha) = \sigma_{\mathrm{eff}}^2 / (2\lambda_1)$.
Further define
\begin{equation}\label{eq:coordinate-r}
\br=(I-\bv\bw^\top)\bx.
\end{equation}
Because $\bw^\top\bv=1$, the exact decomposition
\begin{equation}\label{eq:coordinate2}
\bx=\gamma\bv+\br,
\qquad
\bw^\top\br=0
\end{equation}
holds. As $\alpha\to\alpha_{\mathrm c}^-$, $\br=O_{\mathbb P}(1)$ and $m=O(1)$, whereas $|\gamma|\to\infty$ in probability. Consequently,
\begin{align}
P\bx &= \gamma P\bv+O_{\mathbb P}(1),
\label{eq:centermean1}\\
\bar x &= \gamma\bar v+O_{\mathbb P}(1).
\label{eq:centermean2}
\end{align}
Here and below, $O_{\mathbb P}(1)$ means bounded in probability as $\alpha\to\alpha_{\mathrm c}^-$: for every $\delta_0>0$, the norm of the remainder is below a $\delta_0$-dependent fixed bound with probability at least $1-\delta_0$ for all $\alpha$ sufficiently close to $\alpha_{\mathrm c}$.
\end{lemma}
%

\begin{proof}
The critical coordinate $\zeta=\bw^\top\bz$ itself follows a scalar OU process. Indeed, left-multiplying~\eqref{eq:ou} by $\bw^\top$ and using $\bw^\top M=\lambda_1\bw^\top$ gives
\begin{equation}
\mathrm{d}\zeta
=
-\lambda_1\zeta\,\mathrm{d}t
+
\bw^\top B\,\mathrm{d}\bW.
\end{equation}
Applying It\^o's formula gives
\begin{align}
\mathrm d(\zeta^2)
&=
2\zeta\,\mathrm d\zeta
+
(\mathrm d\zeta)^2 \notag\\
&=
\left(
-2\lambda_1\zeta^2
+
\bw^\top BB^\top\bw
\right)\mathrm dt
+
2\zeta\bw^\top B\,\mathrm d\bW,
\end{align}
where we used
\begin{equation}
(\bw^\top B\,\mathrm d\bW)^2
=
\bw^\top B
(\mathrm d\bW\,\mathrm d\bW^\top)
B^\top\bw
=
\bw^\top BB^\top\bw\,\mathrm dt.
\end{equation}
After integrating over time and taking expectations, the It\^o integral
$\int_0^t 2\zeta(s)\bw^\top B\,\mathrm d\bW(s)$ has zero expectation. Because the stationary mean of $\zeta$ is zero, $\E[\zeta^2]=\Var(\zeta)$, and $\bw^\top BB^\top\bw=\sigma_{\mathrm{eff}}^2$ by~\eqref{eq:sigma_eff-def}, stationarity gives
\begin{equation}
0
=
-2\lambda_1\Var(\zeta)
+
\sigma_{\mathrm{eff}}^2.
\end{equation}
Therefore,
$\Var(\zeta)=\sigma_{\mathrm{eff}}^2/(2\lambda_1)=s^2$.

The projection $\bv\bw^\top$ selects the critical eigendirection, while $I-\bv\bw^\top$ selects the complementary hyperplane $\ker(\bw^\top)$.
Applying $I-\bv\bw^\top$ removes the divergent rank-one covariance term in Lemma~\ref{lem:rank1}; all remaining covariance terms stay bounded under Assumption~\ref{ass:spec}. The deterministic equilibrium, $\bx^*(\alpha)$, is also bounded at the bifurcation. Therefore, $\br=O_{\mathbb P}(1)$ and $m=O(1)$. Because $s\to\infty$ as $\alpha \to \alpha_{\mathrm c}^-$, the probability that $|\gamma|$ lies in any fixed bounded interval tends to zero. Applying $P$ to~\eqref{eq:coordinate2} gives~\eqref{eq:centermean1}. Spatial averaging of~\eqref{eq:coordinate2} gives~\eqref{eq:centermean2}.
\end{proof}

Lemma~\ref{lem:coordinate} is the working picture for the remainder of the paper. It is an eigenvector decomposition, not the decomposition into spatially centered and uniform parts defined by $P$. Close to the transition, $\bx$ consists of a random scalar $\gamma$ that diverges in probability and multiplies the critical right eigenvector $\bv$, plus a remainder that remains bounded in probability. Centered statistics see $P\bv$, whereas the spatial mean sees $\bar v$.

\subsection{Spatial variance}

We now turn to $\Phi$.
Before taking the critical limit, it is useful to keep the exact eigenvector expansion (Property~\ref{prop:modal}). Theorem~\ref{thm:var} identifies the ultimately dominant term.

\begin{proposition}[Exact eigenvector form of the fluctuation term]\label{prop:modal}
Under Assumption~\ref{ass:spec}, define
\begin{equation}\label{eq:Pikl}
\Pi_{k \ell}\equiv\big(\bv^{(\ell)}\big)^\top P\bv^{(k)}
=\big(\bv^{(\ell)}\big)^\top\bv^{(k)}-\frac{(\bone^\top\bv^{(\ell)})(\bone^\top\bv^{(k)})}{N}.
\end{equation}
Then,
\begin{equation}\label{eq:Phimodalgen}
\Phi(\alpha)=\frac{1}{N-1}\sum_{k=1}^{N}\sum_{\ell=1}^{N}\frac{\beta_{k\ell}\Pi_{k\ell}}{\lambda_k+\lambda_{\ell}}.
\end{equation}
\end{proposition}

\begin{proof}
Equation~\eqref{eq:Phimodalgen} follows from~\eqref{eq:Cspectral} and $\Tr(P\bv^{(k)}\bv^{(\ell)\top})=\bv^{(\ell)\top}P\bv^{(k)}$.
\end{proof}

Equation~\eqref{eq:Phimodalgen} has a direct interpretation: every pair of right eigenvectors is weighted by how the noise enters through the corresponding left eigenvectors and by how much of the right-eigenvector pair survives spatial centering. As $\lambda_1\to0^+$, the $(k,\ell)=(1,1)$ term is the only term with a vanishing denominator. Its spatial factor is
\begin{equation}
\Pi_{11}
=
\bv^\top P\bv
=
\|P\bv\|^2.
\end{equation}
Thus, whether the growing critical fluctuation contributes to the spatial variance depends on how much of the critical right eigenvector remains after subtraction of its spatial mean. We introduce the following quantity to express this dependence.

\begin{definition}[Uniformity of the critical right eigenvector]\label{def:rho}
Let $\bv$ be the real critical right eigenvector, normalized by $\|\bv\|=1$. Define
\begin{equation}\label{eq:rho}
\rho
\equiv
\frac{(\bone^\top\bv)^2}{N}
=
N\bar v^{\,2}
\in[0,1].
\end{equation}
\end{definition}

We also define
\begin{equation}\label{eq:rho_c}
\rho_{\mathrm c}
=
\frac{(\bone^\top\bv_{\mathrm c})^2}{N}
=
\lim_{\alpha\to\alpha_{\mathrm c}^-}\rho(\alpha).
\end{equation}

Because $P=I-\bone\bone^\top/N$ and $\|\bv\|=1$, we obtain
\begin{equation}
\|P\bv\|^2
=
1-\rho.
\end{equation}
Equivalently,
\begin{equation}
\Var(\bv)
=
\frac{1-\rho}{N-1}.
\end{equation}
Thus, $1-\rho$ measures the spatial non-uniformity of the critical right eigenvector. We have $\rho=1$ if and only if $\bv=\pm\bone/\sqrt N$, in which case spatial centering removes the critical right eigenvector completely, and $\rho=0$ if and only if the critical right eigenvector has zero spatial mean.

\begin{theorem}[Critical asymptotics of the spatial variance]\label{thm:var}
Under Assumption~\ref{ass:spec},
\begin{equation}\label{eq:Phi(alpha)-final}
\Phi(\alpha) = \frac{\sigma_{\mathrm{eff}}^2}{2\lambda_1}\frac{1-\rho}{N-1}+O(1).
\end{equation}
%
%
\end{theorem}

\begin{proof}
By substituting~\eqref{eq:rank1} in $\Phi(\alpha)=\Tr(PC)/(N-1)$ and using $\Tr(P\bv\bv^\top)=\bv^\top P\bv=1-\rho$, we obtain the leading term in~\eqref{eq:Phi(alpha)-final}. The bounded remainder follows directly from~\eqref{eq:rank1}.
\end{proof}

\begin{remark}\label{rem:no-divergent}
Equation~\eqref{eq:Phi(alpha)-final} shows that $\Phi(\alpha)$ remains bounded when
$\sigma_{\mathrm{eff}}^2(1-\rho)=O(\lambda_1)$. Two simple sufficient conditions are that $B^\top\bw=\boldsymbol0$ for all $\alpha$ sufficiently close to $\alpha_{\mathrm c}$, or that $\bv=\pm\bone/\sqrt N$ for all such $\alpha$.
%
%
\end{remark}

\begin{remark}[Cooperative irreducible systems]\label{rem:pf}
Suppose that the dynamics is cooperative at $\bx^*$, i.e.,
\begin{equation}\label{eq:coop}
\frac{\partial F_i}{\partial x_j}\Big|_{\bx^*}\ge0
\qquad\text{for all }i\ne j,
\end{equation}
and that the associated interaction network is strongly connected. Then $-M$ is an irreducible Metzler matrix. Perron--Frobenius theory~\cite{horn2013matrix} implies, for each $\alpha<\alpha_{\mathrm c}$ sufficiently close to the bifurcation, that the eigenvalue of $M$ with smallest real part is real and algebraically simple and that its right and left eigenvectors can be chosen entrywise positive. This fact explains why the simple-real-eigenvalue setting is natural for cooperative network dynamics. The separation condition $\operatorname{Re}\lambda_k\ge\delta$ for $k\ge2$ in~\eqref{eq:specgapass} and diagonalizability of the non-critical spectrum remain separate assumptions though.

Because the normalized limiting right eigenvector $\bv_{\mathrm c}$ is nonnegative and nonzero, we obtain $\bone^\top\bv_{\mathrm c}>0$ and hence $\rho_{\mathrm c}>0$. If the noise is independent across nodes, i.e., $B=\operatorname{diag}(\sigma_1,\ldots,\sigma_N)$, then $\sigma_{\mathrm{eff,c}}^2>0$ whenever at least one node satisfies $\sigma_i>0$ and $w_{\mathrm c,i}\ne0$. In particular, $\sigma_{\mathrm{eff,c}}^2>0$ holds when every node receives noise of positive strength. Thus, in this common setting, the noise-excitation cancellation (i.e.,  $\sigma_{\mathrm{eff,c}}^2 = 0$) is excluded, while the uniform-eigenvector cancellation (i.e., $\rho_{\mathrm c}=1$) may still occur. Examples of cooperative irreducible systems include the coupled double-well dynamics in~\eqref{eq:doublewell} and the mutualistic, susceptible-infectious-susceptible, and gene-regulatory dynamics considered in the EWS literature~\cite{masuda2024anticipating, maclaren2025applicability, Bandara2026arxiv} in their cooperative parameter regimes.
\end{remark}

One should read Theorem~\ref{thm:var} together with \eqref{eq:decomp} and Proposition~\ref{prop:S}. The structural term, $S(\alpha)$, remains finite and reflects deformation of the equilibrium profile. The fluctuation term, $\Phi(\alpha)$, is the part generated by critical slowing down and is unbounded when
the limiting critical eigendirection is noise-excited (i.e., $\sigma_{\mathrm{eff,c}}^2>0$) and $\bv_{\mathrm c}$ is non-uniform (i.e., $\rho_{\mathrm c}<1$). An observed rise in spatial variance may therefore have either a deterministic or a stochastic origin. Conversely, a decreasing structural term can mask an increasing fluctuation term and make the raw signal move in the wrong direction until very near the transition, as observed in previous studies~\cite{maclaren2025applicability}.

The factor $1-\rho=\|P\bv\|^2$ has a direct geometric meaning. If the critical right eigenvector $\bv_{\mathrm c}$ moves every node equally, centering removes its contribution completely. If different nodes move by different amounts, the non-uniform part remains and contributes to the spatial variance. For comparison, the temporal variance at the $i$th node in a multidimensional OU system satisfies
$C_{ii}= \sigma_{\mathrm{eff}}^2v_i^2 / (2\lambda_1) + O(1)$~\cite{gardiner2009stochastic,Patterson2021AmNat,masuda2024anticipating}.
The leading spatial and temporal variance contributions therefore contain the same inverse-$\lambda_1$ factor, but their amplitudes depend on different summaries of the critical right eigenvector: $(1-\rho)/(N-1)$ for the spatial variance and $v_i^2$ for the temporal variance at node $i$. We call the normalized eigenvector $\bv$ delocalized when its squared norm is spread over $O(N)$ nodes, so a typical entry satisfies $v_i^2=O(N^{-1})$ and no fixed small set of nodes carries an $O(1)$ fraction of the norm. If, in addition, $1-\rho=O(1)$, then both amplitudes of the spatial and temporal variance can be of order $N^{-1}$. For a localized $\bv$, some entries can instead have $v_i^2=O(1)$, and the comparison then depends strongly on which node is observed because the amplitude of the temporal variance can be $O(1)$, $O(1/N)$ or $o(1/N)$ depending on the degree of the eigenvector localization and choice of the node for which the temporal measure is computed.
The same distinction arises in continuous media when the linearization of stochastic partial differential equations is spatially heterogeneous~\cite{Bernuzzi2025JDDE}.

Within the one-dimensional critical setting of Assumption~\ref{ass:spec}, the rate at which $\lambda_1$ vanishes is determined by the bifurcation type~\cite{kuznetsov2004elements}.

\begin{corollary}[Variance exponents]\label{cor:exponents}
Suppose $\sigma_{\mathrm{eff}}^2(1-\rho)$ tends to a positive limit and write $\varepsilon=\alpha_{\mathrm c}-\alpha$.
\begin{enumerate}[label=(\roman*), leftmargin=*, itemsep=1pt]
\item At a generic saddle-node bifurcation, $\lambda_1=\kappa\sqrt{\varepsilon}(1+o(1))$, so $\E[V]$ grows in proportion to $\varepsilon^{-1/2}$.
\item At a nondegenerate transcritical or pitchfork bifurcation, $\lambda_1=\kappa\varepsilon(1+o(1))$, so $\E[V]$ grows in proportion to $\varepsilon^{-1}$.
\end{enumerate}
\end{corollary}

\subsection{Coefficient of variation}

It is natural to reason that $\mathrm{CV}=\sqrt V/|\bar x|$ should inherit half of the variance exponent because the equilibrium mean remains finite at a steady-state bifurcation. That reasoning is valid only when the random term $\gamma\bar v$ in the spatial mean is small compared with the remaining terms. Equations~\eqref{eq:centermean1} and~\eqref{eq:centermean2} show that, sufficiently near the transition, both the centered snapshot $P\bx$ and the spatial mean $\bar x$ are governed by the same random scalar $\gamma$. Because $V=\|P\bx\|^2/(N-1)$, both $\sqrt V$ and $|\bar x|$ usually contain the factor $|\gamma|$, which cancels in their ratio. The theorem below makes this statement precise.

\begin{theorem}[Asymptotics of the spatial coefficient of variation]\label{thm:cv}
Assume $\sigma_{\mathrm{eff,c}}^2>0$ and that $\Pr(\bar x=0)=0$ for $\alpha<\alpha_{\mathrm c}$ sufficiently close to $\alpha_{\mathrm c}$. For the CV defined in~\eqref{eq:CV}, the following convergence statements hold in probability within the linearized model:
\begin{enumerate}[label=(\roman*), leftmargin=*, itemsep=1pt]
\item If $0<\rho_{\mathrm c}<1$, then
\begin{equation}\label{eq:cvlimit}
\mathrm{CV}\longrightarrow\sqrt{\frac{N(1-\rho_{\mathrm c})}{(N-1)\rho_{\mathrm c}}}.
\end{equation}
\item If $\rho_{\mathrm c}=1$, then $\mathrm{CV}\to0$.
\item If $\rho_{\mathrm c}=0$, then $\mathrm{CV}\to\infty$ in probability.
\end{enumerate}
\end{theorem}

\begin{proof}
Equations~\eqref{eq:centermean1} and~\eqref{eq:centermean2}, together with $V = \|P\bx\|^2/(N-1)$ and $|\gamma|\to\infty$ in probability, give
\begin{equation}\label{eq:CV-law}
\mathrm{CV}
=
\frac{
\|P\bv\|/\sqrt{N-1}+o_{\mathbb P}(1)
}{
|\bar v+o_{\mathbb P}(1)|
}.
\end{equation}
If $0<\rho_{\mathrm c}<1$, then $\|P\bv\|\to\sqrt{1-\rho_{\mathrm c}}$ and $|\bar v|\to\sqrt{\rho_{\mathrm c}/N}>0$, so~\eqref{eq:cvlimit} follows from~\eqref{eq:CV-law}. If $\rho_{\mathrm c}=1$, then the numerator in~\eqref{eq:CV-law} converges to zero while the denominator converges to $1/\sqrt N$, and hence $\mathrm{CV}\to0$. If $\rho_{\mathrm c}=0$, then the numerator converges to $1/\sqrt{N-1}$ while the denominator converges to zero in probability, and therefore $\mathrm{CV}\to\infty$ in probability.
\end{proof}

\begin{remark}
Under an unconstrained Gaussian approximation, $\bar{x}$, which is the denominator of $\mathrm{CV}$, has a continuous density near zero, so moments such as $\E[\mathrm{CV}]$ need not exist. We therefore state the ultimate result in probability rather than through potentially unstable moments.
\end{remark}
 
For the cooperative irreducible dynamical systems described in Remark~\ref{rem:pf}, we obtain $\rho_{\mathrm c}>0$. Therefore, the divergent case $\rho_{\mathrm c}=0$ in Theorem~\ref{thm:cv} is excluded. The limiting $\mathrm{CV}$ value is small when the critical right eigenvector is nearly uniform and large when its spatial mean is small relative to its centered variation.

This conclusion does not rule out a practically important intermediate regime of $\alpha$. If the stochastic term $\gamma P\bv$ already dominates the centered snapshot $P\bx$ and hence $V$, but the corresponding term $\gamma\bar v$ in $\bar x$ is still small compared with the deterministic equilibrium mean, then the denominator is approximately fixed and the CV grows roughly like $\sqrt V$. As the bifurcation is approached further, the term $\gamma\bar v$ becomes important in $\bar x$, and the growth of the CV crosses over to the saturation described by Theorem~\ref{thm:cv}.

\subsection{Skewness, kurtosis, and Moran's $I$ do not diverge}\label{sub:skewness-no-diverge}

Skewness, kurtosis, and Moran's $I$ differ from both the variance and the CV in a decisive way: they depend only on the relative pattern of the deviations from the spatial mean, not on their overall magnitude or on the spatial mean itself. Indeed, multiplying all centered deviations $x_i-\bar x$ by the same nonzero scalar, $\tilde{c}$, leaves $g_2$ and $I_{\mathrm M}$ unchanged, while $g_1$ changes only by the sign of $\tilde{c}$.
By contrast, $V$ and $\mathrm{CV}$ are multiplied by $\tilde{c}^2$ and $|\tilde{c}|$, respectively. Critical slowing down primarily increases the magnitude of the random coefficient $\gamma$ multiplying the critical right eigenvector; a statistic that divides out this common amplitude does not inherit the variance divergence.

\begin{theorem}[Saturation of scale-invariant spatial statistics]\label{thm:ratios}
Assume $\sigma_{\mathrm{eff,c}}^2>0$ and $\rho_{\mathrm c}<1$. With $\gamma$ from Lemma~\ref{lem:coordinate},
\begin{equation}\label{eq:g1lim}
g_1
=
\operatorname{sgn}(\gamma)
\frac{\mu_3(\bv_{\mathrm c})}{\mu_2(\bv_{\mathrm c})^{3/2}}
+
o_{\mathbb P}(1),
\end{equation}
\begin{equation}\label{eq:g2lim}
g_2
\longrightarrow
\frac{\mu_4(\bv_{\mathrm c})}{\mu_2(\bv_{\mathrm c})^2}
\end{equation}
in probability, and
\begin{equation}\label{eq:moranlim}
I_{\mathrm M}
\longrightarrow
\frac{N}{W}
\frac{\bv_{\mathrm c}^{\top}PAP\bv_{\mathrm c}}{\bv_{\mathrm c}^{\top}P\bv_{\mathrm c}}
\end{equation}
in probability.
\end{theorem}

\begin{remark}
The limiting magnitudes are finite and independent of $\lambda_1$ and the overall noise scale.
\end{remark}

\begin{proof}
Equation~\eqref{eq:centermean1} reads
$P\bx=\gamma P\bv+\boldsymbol{\eta}$ with
$\boldsymbol{\eta}=O_{\mathbb P}(1)$.
Using this equation, we expand each term of
\begin{equation}
\mu_k(\bx)
=
\frac{1}{N}\sum_{i=1}^N (P\bx)_i^k
\end{equation}
as
\begin{equation}
(P\bx)_i^k
=
\big[\gamma(P\bv)_i+\eta_i\big]^k
=
\gamma^k(P\bv)_i^k
+
\sum_{j=1}^k
\binom{k}{j}
\gamma^{k-j}(P\bv)_i^{k-j}\eta_i^j.
\end{equation}
For fixed $N$ and $k$, the components of $P\bv$ are bounded and $\eta_i=O_{\mathbb P}(1)$. Because Lemma~\ref{lem:coordinate} gives $|\gamma|\to\infty$ in probability, every term in the sum is $O_{\mathbb P}(|\gamma|^{k-1})$. Therefore, we obtain
\begin{equation}
\mu_k(\bx)
=
\gamma^k\frac{1}{N}\sum_{i=1}^N(P\bv)_i^k
+
O_{\mathbb P}(|\gamma|^{k-1})
=
\gamma^k\mu_k(\bv)
+
O_{\mathbb P}(|\gamma|^{k-1}).
\end{equation}
Because $\bv\to\bv_{\mathrm c}$ and $\mu_2(\bv_{\mathrm c})>0$ (the latter holds because $\rho_{\mathrm c} \neq 1$, which implies $\bv_{\mathrm c} \neq \pm \bone/\sqrt{N}$), the powers of $\gamma$ cancel in the normalized moment ratios, giving~\eqref{eq:g1lim} and~\eqref{eq:g2lim}. Similarly, substituting $P\bx=\gamma P\bv+O_{\mathbb P}(1)$ into the numerator and denominator of Moran's $I$ shows that both quadratic forms have leading order $\gamma^2$; their common factor cancels, and $\bv\to\bv_{\mathrm c}$ gives~\eqref{eq:moranlim}.
\end{proof}

None of these spatial EWSs diverges as $\alpha \to \alpha_{\mathrm c}^-$.
We note that convergence by itself does not make a statistic useless. Temporal lagged autocorrelation is a standard EWS although it approaches $1$ as $\alpha \to \alpha_{\mathrm c}^-$ \cite{Scheffer2009Nature}. What makes a converging statistic interpretable is a known direction of approach and a sufficiently stable estimator. The three spatial EWSs here do not provide a universal direction: two networks can have critical right eigenvectors with opposite skewness, different kurtosis, or different Moran's $I$. The problem is therefore not merely the absence of divergence, but the absence of a network-independent and ideally dynamics-independent rule saying whether an increase or a decrease is the warning sign.

\begin{corollary}[Sign instability of spatial skewness]\label{cor:signflip}
Under the hypotheses of Theorem~\ref{thm:ratios}, suppose that
$\mu_2(\bv_{\mathrm c})>0$, and define
\begin{equation}
K_{1,\mathrm c}
=
\frac{\mu_3(\bv_{\mathrm c})}
{\mu_2(\bv_{\mathrm c})^{3/2}}.
\end{equation}
Let $F_{\mathcal{N}}$ denote the cumulative distribution function of a standard normal random variable. Then
\begin{equation}
\Pr(\gamma>0)
=
F_{\mathcal{N}}\!\left(\frac{m(\alpha)}{s(\alpha)}\right)
\longrightarrow
\frac{1}{2}.
\end{equation}
(Recall that $m(\alpha)$ and $s(\alpha)$ are defined in Lemma~\ref{lem:coordinate}.)
Consequently, $g_1$ converges in distribution to a random variable taking the values $K_{1,\mathrm c}$ and $-K_{1,\mathrm c}$ with probability $1/2$ each. In particular, $\E[g_1]\longrightarrow0$, and $|g_1|\longrightarrow|K_{1,\mathrm c}|$
in probability.
\end{corollary}

\begin{proof}
The convergence $\Pr(\gamma>0)\to1/2$ holds because $m=O(1)$ and $s\to\infty$. 
This convergence and \eqref{eq:g1lim} give the signed two-point distributional limit. The sample skewness is bounded for fixed $N$, so the convergence also gives $\E[g_1]\to0$. The convergence of $|g_1|$ follows from the continuous mapping theorem.
\end{proof}

The corollary does not imply that the snapshot itself takes only two possible forms. From Lemma~\ref{lem:coordinate}, we recall
\begin{equation}
\bx
=
\gamma\bv+\br,
\end{equation}
where
\begin{equation}
\br
=
(I-\bv\bw^\top)\bx^*
+
(I-\bv\bw^\top)\bz.
\label{eq:br-projected-finite}
\end{equation}
The first term on the right-hand side of~\eqref{eq:br-projected-finite} is a bounded contribution from the deterministic equilibrium, whereas the second term is the bounded stochastic contribution associated with the non-critical eigendirections. At a finite distance from the bifurcation, these terms perturb the centered snapshot and can produce a continuous range of skewness values. As the transition is approached, however, $|\gamma|$ grows while $\br$ remains bounded in probability, so $P\bx$ is increasingly well approximated by $\gamma P\bv$. Therefore, snapshots with $\gamma>0$ have skewness close to $K_{1,\mathrm c}$ and those with $\gamma<0$ have skewness close to $-K_{1,\mathrm c}$, although $\bx$ itself remains continuously distributed.
This signed two-point limit is specific to skewness because skewness changes sign when all centered deviations are multiplied by $-1$. Kurtosis and Moran's $I$ are unchanged by this sign reversal, so they instead converge to single network-dependent constants.

A caveat applies to the asymptotic limits in this subsection and to~\eqref{eq:cvlimit}. They describe the regime in which the critical contribution $\gamma P\bv$ dominates the bounded remainder in~\eqref{eq:centermean1}. At fixed nonzero noise, the Ornstein--Uhlenbeck approximation may cease to be accurate before this ultimate regime is reached, because large excursions can leave the local basin of the equilibrium. Whether the regime is observable therefore depends on the noise strength, basin geometry, and rate of parameter variation. At some distances from the bifurcation where tipping does not yet occur in the presence of dynamical noise, the equilibrium profile $\bx^*(\alpha)$ and non-critical fluctuations may still have substantial effects, so the observed CV, skewness, kurtosis, and Moran's $I$ need not be close to the limits in~\eqref{eq:cvlimit}, \eqref{eq:g1lim}, \eqref{eq:g2lim}, and \eqref{eq:moranlim}. Likewise, the two-point skewness distribution in Corollary~\ref{cor:signflip} becomes visible only when $s(\alpha)$ is large relative to $|m(\alpha)|$; when $|m(\alpha)|$ dominates, the sign of $\gamma=m+\zeta$ is nearly deterministic and only one of the two skewness signs is typically observed. 
Like the asymptotic laws~\eqref{eq:cvlimit}, \eqref{eq:g1lim}, \eqref{eq:g2lim}, and \eqref{eq:moranlim} do not,
these pre-asymptotic properties do not create a network-independent warning direction for these four spatial statistics.

\section{Homogeneous networks}\label{sec:homog}

Theorem~\ref{thm:var} shows that the potentially divergent contribution of the critical right eigenvector to the spatial variance is proportional to $1-\rho=\|P\bv\|^2$. Exact cancellation occurs when the critical right eigenvector is uniform throughout a neighborhood of the bifurcation, $\bv=\pm\bone/\sqrt N$, so that $\rho=1$ and spatial centering removes this eigendirection completely. This situation arises naturally on a uniform equilibrium branch of a homogeneous network, which is the setting used in much of the classical spatial EWS literature on regular lattices~\cite{dakos2011slowing, sankaran2018implications} or their continuous counterparts~\cite{Gowda2015CNSNS, Clarke2026JPhysComplexity, tirabassi2024linear}.
We now identify a broad class of systems for which this situation arises naturally.

In the next theorem and throughout, for a function $G(u_1,u_2)$ of two arguments, we write $\partial_1 G$ and $\partial_2 G$ for its partial derivatives with respect to the first and the second argument, respectively.

\begin{theorem}[Uniform critical-eigenvector cancellation]\label{thm:lattice}
Let $A$ be the adjacency matrix of a connected undirected regular network with degree $d$, and consider the class of dynamics given by
\begin{equation}
F_i(\bx;\alpha)
=
f(x_i;\alpha)
+
D\sum_{j=1}^N A_{ij}G(x_i,x_j),
\label{eq:canonical-dyn}
\end{equation}
with homogeneous noise $B=\sigma I$. Let $x^*(\alpha)$ be a branch of solutions of
\begin{equation}
f(x^*(\alpha);\alpha)
+
Dd G(x^*(\alpha),x^*(\alpha))
=
0.
\end{equation}
Then, $\bx^*(\alpha)=x^*(\alpha)\bone$ is a uniform equilibrium branch.

Suppose that this branch is asymptotically stable for $\alpha<\alpha_{\mathrm c}$ and loses stability as $\alpha\to\alpha_{\mathrm c}^-$, and that

\begin{equation}
D\,\partial_2G(x^*(\alpha),x^*(\alpha))>0
\end{equation}
near $\alpha_{\mathrm c}$, with
\begin{equation}
D\,\partial_2G(x^*(\alpha),x^*(\alpha))
\longrightarrow
\omega_{\mathrm c}>0.
\end{equation}
Then, $S(\alpha)=0$, the critical eigenspace is $\operatorname{span}\{\bone\}$, and the normalized critical right eigenvector may be chosen as $\bv=\bone/\sqrt N$ throughout this branch. Consequently, $\rho=1$ and $\bv_{\mathrm c}=\bone/\sqrt N$ under this orientation.

Let $\mu_1=d>\mu_2\ge\cdots\ge\mu_N$ be the eigenvalues of $A$, and let $\lambda_k$ denote the corresponding eigenvalues of the stability matrix $M$. Then
\begin{equation}\label{eq:latticeV}
\E[V]
=
\frac{\sigma^2}{2(N-1)}
\sum_{k=2}^{N}\frac{1}{\lambda_k}.
\end{equation}
Moreover,
\begin{equation}\label{eq:latticelimit}
\E[V]
\longrightarrow
\frac{\sigma^2}{2\omega_{\mathrm c}(N-1)}
\sum_{k=2}^{N}\frac{1}{\mu_1-\mu_k}
<\infty.
\end{equation}
\end{theorem}

\begin{proof}
Because the network is regular,
$\sum_{j=1}^N A_{ij}=d$
for each $i$. Therefore, substitution of $\bx^*=x^*\bone$ into the deterministic dynamics gives
$F_i(x^*\bone;\alpha)
=
f(x^*;\alpha)
+
Dd G(x^*,x^*)$,
which vanishes by assumption. Therefore, $\bx^*=x^*\bone$ is indeed an equilibrium branch, and its uniformity gives $S(\alpha)=0$.

Along this branch, define
\begin{equation}
a(\alpha)
=
-\partial_x f(x^*(\alpha);\alpha)
-
Dd\,\partial_1G(x^*(\alpha),x^*(\alpha))
\end{equation}
and
\begin{equation}
\omega(\alpha)
=
D\,\partial_2G(x^*(\alpha),x^*(\alpha)).
\end{equation}
Differentiating the network dynamics at the uniform equilibrium gives
$M=aI-\omega A$.
Hence $M$ and $A$ have the same eigenvectors, and their eigenvalues are related by
$\lambda_k=a-\omega\mu_k$.

For a connected undirected regular network, $\mu_1=d$ is simple and its normalized eigenvector is $\bone/\sqrt N$. Because $\omega>0$, eigenvalue $\lambda_1=a-\omega\mu_1$ is the smallest eigenvalue of $M$. Therefore, as the stable uniform branch loses stability, we can choose the critical right eigenvector to tend to $\bone/\sqrt N$, so $\rho=1$.

Because $A$ and hence $M$ are symmetric and $B=\sigma I$, we obtain
$C=\frac{\sigma^2}{2}M^{-1}$.
Because $M$ is symmetric, the right eigenvectors
$\bv^{(1)},\ldots,\bv^{(N)}$
introduced above can be chosen as an orthonormal basis, and the corresponding left eigenvectors coincide with them. We choose
$\bv^{(1)}=\bone/\sqrt N$. Then, we obtain
\begin{equation}
C
=
\frac{\sigma^2}{2}
\sum_{k=1}^N
\frac{1}{\lambda_k}
\bv^{(k)}
\bv^{(k)\top}.
\end{equation}
Because
$P\bv^{(1)}=\boldsymbol0$
and
$P\bv^{(k)}=\bv^{(k)}$
for $k\ge2$, we obtain
\begin{equation}
\Tr(PC)
=
\frac{\sigma^2}{2}
\sum_{k=2}^N
\frac{1}{\lambda_k}.
\end{equation}
Because $S(\alpha)=0$, the decomposition~\eqref{eq:decomp} gives~\eqref{eq:latticeV}.
Finally, $\lambda_k=a-\omega\mu_k$ implies that
$\lambda_k
=
\lambda_1+\omega(\mu_1-\mu_k)$, $k \in \{2, \ldots, N \}$.
Because $\lambda_1\to0$ and $\omega\to\omega_{\mathrm c}$ as $\alpha \to \alpha_{\mathrm c}^-$, we obtain
$\lambda_k
\longrightarrow
\omega_{\mathrm c}(\mu_1-\mu_k)>0$,
$k \in \{2,\ldots,N \}$,
which gives~\eqref{eq:latticelimit}.

\end{proof}

The cancellation in Theorem~\ref{thm:lattice} is not restricted to dynamics of the form given by \eqref{eq:canonical-dyn}.
Another common convention is that each node responds to its aggregate network input~\cite{Thibeault2024NatPhys}. The next corollary is a variant of
Theorem~\ref{thm:lattice} for the latter class of dynamics on networks.

\begin{corollary}[Uniform cancellation for aggregate-input dynamics]\label{cor:inputcoupling}
Let $A$ be the adjacency matrix of a connected undirected regular graph with degree $d$, and let $B=\sigma I$. Consider dynamics on networks given by
\begin{equation}\label{eq:inputcoupling}
F_i(\bx;\alpha)
=
H\!\left(x_i,\sum_{j=1}^N A_{ij}x_j;\alpha\right),
\qquad
i \in \{ 1,\ldots,N \},
\end{equation}
where $H$ is continuously differentiable in its first two arguments. Suppose that $x^*(\alpha)$ is a branch of solutions of
\begin{equation}\label{eq:input-uniform-equilibrium}
H\!\left(x^*(\alpha),dx^*(\alpha);\alpha\right)=0.
\end{equation}
Then, $\bx^*(\alpha)=x^*(\alpha)\bone$ is a uniform equilibrium branch.

Suppose that this branch is asymptotically stable for $\alpha<\alpha_{\mathrm c}$ and loses stability as $\alpha\to\alpha_{\mathrm c}^-$, and suppose
$\partial_2H\!\left(x^*(\alpha),dx^*(\alpha);\alpha\right)>0$
near $\alpha_{\mathrm c}$, with
$\partial_2H\!\left(x^*(\alpha),dx^*(\alpha);\alpha\right)
\longrightarrow
\omega_{\mathrm c}>0$.
Then, the critical eigenspace is $\operatorname{span}\{\bone\}$, and the normalized critical right eigenvector may be chosen as $\bv=\bone/\sqrt N$ throughout this branch. Consequently, $\rho=1$ and $\bv_{\mathrm c}=\bone/\sqrt N$ under this orientation.
Moreover, the conclusions~\eqref{eq:latticeV} and~\eqref{eq:latticelimit} of Theorem~\ref{thm:lattice} hold.
\end{corollary}

\begin{proof}
Because the network is regular, we obtain $A\bone=d\bone$. Therefore, substituting $\bx^*=x^*\bone$ into~\eqref{eq:inputcoupling} gives
$F_i(x^*\bone;\alpha)
=
H(x^*,dx^*;\alpha)$,
which vanishes by~\eqref{eq:input-uniform-equilibrium}. Thus, $\bx^*=x^*\bone$ is an equilibrium branch, and its uniformity gives $S(\alpha)=0$.

Along this branch, define
$a(\alpha)
=
-\partial_1H\!\left(x^*(\alpha),dx^*(\alpha);\alpha\right)$
and
$\omega(\alpha)
=
\partial_2H\!\left(x^*(\alpha),dx^*(\alpha);\alpha\right)$.
Differentiating~\eqref{eq:inputcoupling} at the uniform equilibrium gives
$M=aI-\omega A$.
The remainder of the argument is the same as in the proof of Theorem~\ref{thm:lattice}.
\end{proof}

On the uniform branch, we obtain $S(\alpha) \equiv 0$, so the spatial variance is a pure fluctuation statistic. In contrast, on a heterogeneous equilibrium branch, $S(\alpha)$ is generally nonzero even far from the transition and can increase or decrease as the control parameter changes. A statistic developed on a lattice therefore omits the deterministic background that can dominate when the same method is transferred to a heterogeneous network.

\section{Uncertainty of a spatial early warning signal}\label{sec:uncertainty}

A large expected signal is useful only relative to its snapshot-to-snapshot variability~\cite{masuda2024anticipating, yu2026covariance}. Spatial sampling offers $N$ observations at once, suggesting that a statistic should become precise as the network grows. However, near a bifurcation, the divergent part of all node fluctuations is generated by one common random coordinate. The $N$ nodes are then not independent, which causes the uncertainty of the spatial variance to remain persistent.

To study the snapshot-to-snapshot uncertainty of the spatial variance, we return to the decomposition used in the proof of Proposition~\ref{prop:decomp}, but now before taking the expectation. For an individual snapshot, we obtain
\begin{equation}\label{eq:V-realization-decomp}
V
=
S(\alpha)
+
\frac{2(P\bx^*)^\top P\bz}{N-1}
+
\frac{\bz^\top P\bz}{N-1}.
\end{equation}
The structural term $S(\alpha)$ is deterministic, whereas the second and third terms vary across stochastic realizations due to the stochastic fluctuation $\bz$. The cross term has zero expectation, while the last term is the spatial variance generated by $\bz$. We denote this last term by
\begin{equation}\label{eq:Vfluc}
V_{\mathrm{fluc}}
=
\frac{\bz^\top P\bz}{N-1}.
\end{equation}
Note that Proposition~\ref{prop:decomp} implies that
\begin{equation}\label{eq:meanVfluc}
\E[V_{\mathrm{fluc}}]
= \Phi(\alpha) =
\frac{\Tr(PC)}{N-1}.
\end{equation}
In the following text, we first quantify the uncertainty of $V_{\mathrm{fluc}}$ and then return to the full spatial variance $V$.

\begin{theorem}[Persistent relative uncertainty of spatial variance]\label{thm:varuncertainty}
For the Gaussian fluctuation vector $\bz$, we obtain
\begin{equation}\label{eq:varVfluc}
\Var[V_{\mathrm{fluc}}]
=
\frac{2\Tr(PCPC)}{(N-1)^2}.
\end{equation}
If $\sigma_{\mathrm{eff,c}}^2>0$ and $\rho_{\mathrm c}<1$, then
\begin{equation}\label{eq:relativeVfluc}
\frac{\sd[V_{\mathrm{fluc}}]}
{\E[V_{\mathrm{fluc}}]}
\longrightarrow \sqrt{2}.
\end{equation}
The same limit holds for the full spatial variance $V$, i.e., $\sd[V] / \E[V] \longrightarrow \sqrt{2}$.
\end{theorem}

We prove Theorem~\ref{thm:varuncertainty} in Appendix~\ref{app:varuncertainty}.

The limiting value $\sqrt2$ is the same for every fixed $N$. Thus, within the rank-one critical approximation, observing more nodes does not produce the usual $N^{-1/2}$ reduction in the relative uncertainty of $V$. This is because the same random critical amplitude is shared across the nodes. This result contrasts with temporal-variance estimation, which improves when many observations are collected over time and the spacing between consecutive observations is large enough to make their temporal correlation negligible.

The CV requires a separate treatment because its denominator, $\bar{x}$, is itself random.
The uncertainty of the CV is different from that of the spatial variance. By Theorem~\ref{thm:cv}, when $\rho_{\mathrm c}>0$, the CV converges in probability to the finite deterministic limit
$\sqrt{N(1-\rho_{\mathrm c})/[(N-1)\rho_{\mathrm c}]}$, so its snapshot-to-snapshot variation becomes negligible for typical realizations sufficiently near the transition. However, the CV's ordinary moments and quantities derived from them, such as its variance, are generally not useful. At any fixed $\alpha<\alpha_{\mathrm c}$ for which $\Var(\bar x)>0$, the spatial mean $\bar x$ has a Gaussian density and can therefore be arbitrarily close to zero. Therefore, $\E[\mathrm{CV}]=\infty$ and hence $\Var[\mathrm{CV}]$ is not defined. Thus, for the CV, convergence in probability rather than its variance provides the appropriate description of the asymptotic snapshot uncertainty.

The uncertainties of the three scale-invariant spatial EWSs are also qualitatively different.

\begin{theorem}[Uncertainty of skewness, kurtosis, and Moran's $I$]\label{thm:uncertainty}
Under the hypotheses of Theorem~\ref{thm:ratios}, let
\begin{equation}
K_{1,\mathrm c}
=
\frac{\mu_3(\bv_{\mathrm c})}{\mu_2(\bv_{\mathrm c})^{3/2}}.
\end{equation}
Then
\begin{enumerate}[label=(\roman*), leftmargin=*, itemsep=2pt]
\item $\sd[g_1]\to |K_{1,\mathrm c}|$.
\item $g_2$ and $I_{\mathrm M}$ converge in $L^2$ to the deterministic constants in~\eqref{eq:g2lim} and~\eqref{eq:moranlim}, so their variances tend to zero.
\end{enumerate}
\end{theorem}

\begin{proof}
For fixed $N$, the sample skewness is bounded. Therefore, its signed two-point distributional limit, taking $K_{1,\mathrm c}$ and $-K_{1,\mathrm c}$ each with probability $1/2$, also gives convergence of its first two moments and limiting variance $K_{1,\mathrm c}^2$. The bounds $1\le g_2\le N$ and $|I_{\mathrm M}|\le(N/W)\|A\|_2$ provide uniform integrability, so convergence in probability upgrades to convergence in $L^2$.
\end{proof}

The contrast among the EWSs is now clear. The relative standard deviation of the spatial variance tends to $\sqrt2$. When $\rho_{\mathrm c}>0$, the CV becomes concentrated in probability around the deterministic limit in Theorem~\ref{thm:cv}, but its ordinary variance is generally not finite because of rare realizations with spatial mean arbitrarily close to zero. Skewness approaches a fixed magnitude but not a fixed sign, so its uncertainty can be as large as the signal itself. Kurtosis and Moran's $I$ become asymptotically deterministic, but the constants they approach are network dependent and do not supply a universal warning direction. Their limiting precision should also not be confused with rapid convergence: if the total non-critical variance is large, the bounded remainder in Lemma~\ref{lem:coordinate} can remain important until the dynamical system is fairly close to the bifurcation.

\section{Baseline referencing}\label{sec:baseline}

Baseline referencing is a preprocessing step for computing spatial EWSs. The idea is to reference each $x_i$ to a running mean of $x_i$, denoted by $b_i$ (defined by \eqref{eq:baselinedef} below), computed from a small number of early control parameter values \cite{Bandara2026arxiv}. 
Specifically, we use either $x_i-b_i$ (subtractive variant) or $x_i/b_i$ (ratio variant) to compute any spatial EWS. For example, the subtractive variant of the spatial variance is the sample variance of $\{ x_1 - b_1, \ldots, x_N - b_N \}$. 
Baseline referencing aims to reduce persistent node-to-node heterogeneity that is not generated by critical slowing down, and it numerically outperforms the original spatial EWSs in several settings \cite{Bandara2026arxiv}.

The ratio variant is natural when node states carry different units or different natural scales, because multiplying $x_i$ and $b_i$ by the same node-specific factor leaves $x_i/b_i$ unchanged. Its analysis runs parallel to the subtractive one, and we defer it to Appendix~\ref{app:ratio-baseline-referencing}. The remainder of this section treats the subtractive variant, which is the one that admits an exact and interpretable decomposition.

Let the baseline for the subtractive variant be the average of $\ell$ independent stationary snapshots at a fixed reference parameter $\alpha_0$:
\begin{equation}\label{eq:baselinedef}
\bb=\frac1\ell\sum_{r=1}^{\ell}\bx^{(r)}(\alpha_0)=\bx^*(\alpha_0)+\boldsymbol\epsilon,
\end{equation}
where $\boldsymbol\epsilon$ has distribution $\mathcal{N}\left( \boldsymbol 0, \frac{C(\alpha_0)}{\ell} \right)$.
The baseline samples are independent of the current snapshot, and $\alpha_0$ is typically farther from the bifurcation point than the current $\alpha$ value is. Define
\begin{align}
\by^\Delta &= \bx-\bb,\\
\bde(\alpha) &= \bx^*(\alpha)-\bx^*(\alpha_0),\\
V_\Delta &= \Var(\by^\Delta).
\end{align}
Note that $V_\Delta$ is the spatial variance after the subtractive baseline referencing.
A baseline assembled from $\ell$ samples obtained at several nearby reference parameters has the same structure after replacing $\bx^*(\alpha_0)$ and $C(\alpha_0)/\ell$ by the corresponding averages. For example, if we sample one independent stationary snapshot at each $\alpha = \alpha_i$, $i \in \{ 1, \ldots, \ell \} $, then
we replace $\bx^*(\alpha_0)$ and $C(\alpha_0)/\ell$ by $\sum_{i=1}^{\ell} \bx^*(\alpha_i) / \ell$ and $\sum_{i=1}^{\ell} C(\alpha_i) / \ell^2$, respectively.

\begin{proposition}[Decomposition of the baselined variance]\label{prop:baseline}
Under the above sampling scheme,
\begin{equation}\label{eq:EVdelta}
\E[V_\Delta] = \Var(\bde(\alpha))+\Phi(\alpha)+\frac1\ell\Phi(\alpha_0).
\end{equation}
\end{proposition}

\begin{proof}
We have
\begin{equation}
P\by^\Delta = P\bx - P\bb = P(\bx^*(\alpha) + \bz) - P(\bx^*(\alpha_0)+\boldsymbol\epsilon) = P\bde+P\bz-P\boldsymbol\epsilon.
\label{eq:baseline-referencing-proof-1}
\end{equation}
Substitution of \eqref{eq:baseline-referencing-proof-1} in $V_\Delta=\|P\by^\Delta\|^2/(N-1)$ yields
\begin{align}
V_\Delta
&=
\frac{1}{N-1} \left[ \|P\bde\|^2
+
\|P\bz\|^2
+
\|P\boldsymbol\epsilon\|^2 \right.
\notag\\
&\quad
\left. +
2(P\bde)^\top P\bz
-
2(P\bde)^\top P\boldsymbol\epsilon
-
2(P\bz)^\top P\boldsymbol\epsilon \right].
\label{eq:baseline-referencing-proof-2}
\end{align}
Because
the vectors $\bz$ and $\boldsymbol\epsilon$ are centered,
$\E[\bz] = \E[\boldsymbol\epsilon] = 0$ so that the expectation of the fourth and fifth terms on the right-hand side of \eqref{eq:baseline-referencing-proof-2} vanishes.
Moreover, because $\bz$ and $\boldsymbol\epsilon$ are independent, the expectation of the last term also vanishes.
The first term is equal to $\Var(\bde(\alpha))$.
Because $P^\top=P$ and $P^2=P$,
\begin{equation}
\E[\|P\bz\|^2]
=
\E[\bz^\top P\bz]
=
\Tr\bigl(PC(\alpha)\bigr).
\label{eq:baseline-referencing-proof-3}
\end{equation}
Finally, we obtain
\begin{equation}
\E[\|P\boldsymbol\epsilon\|^2]
=
\Tr\bigl(P\operatorname{Cov}(\boldsymbol\epsilon)\bigr)
=
\frac{1}{\ell}\Tr\bigl(PC(\alpha_0)\bigr).
\label{eq:baseline-referencing-proof-4}
\end{equation}
Applying the definition of $\Phi$ to
\eqref{eq:baseline-referencing-proof-3} and
\eqref{eq:baseline-referencing-proof-4}
concludes the proof.
\end{proof}

Comparing~\eqref{eq:EVdelta} with~\eqref{eq:decomp} shows what the baseline referencing changes.
First, the current fluctuation term $\Phi(\alpha)$, including its critical divergence, is untouched.
Second, at the reference parameter $\alpha=\alpha_0$, the first term on the right-hand side of~\eqref{eq:EVdelta} vanishes because
$\bde(\alpha_0)
=
\bx^*(\alpha_0)-\bx^*(\alpha_0)
=
\boldsymbol0$, and hence $\Var(\bde(\alpha_0))=0$.
For the raw spatial variance evaluated at the same parameter value, however, the corresponding deterministic contribution in~\eqref{eq:decomp} is
$S(\alpha_0)=\Var(\bx^*(\alpha_0))$,
which is generally nonzero. Thus, baselining replaces the pre-existing node-to-node variation of the equilibrium profile at $\alpha = \alpha_0$ (i.e., $S(\alpha_0)$) by zero. The expected baselined variance at $\alpha_0$ therefore contains only the current fluctuation contribution 
(i.e., $\Phi(\alpha)$)
and the additional fluctuation caused by estimating the baseline from finitely many snapshots
(i.e., $\Phi(\alpha_0)/\ell$).
As $\alpha$ moves away from $\alpha_0$, however, $\Var(\bde(\alpha))$ may become appreciably positive.
Baselining therefore does not, by itself, guarantee an improvement of the spatial variance as an EWS. What it guarantees is a zero deterministic offset at $\alpha=\alpha_0$, which may remain approximately small as $\alpha$ varies. This mechanism helps explain the improved numerical performance of the baseline-referenced spatial variance reported in~\cite{Bandara2026arxiv} and is often the main practical benefit of baseline referencing.
Third, \eqref{eq:EVdelta}
also has the constant term $\Phi(\alpha_0)/\ell$. This term shifts the level of the expected variance but does not change its trend with $\alpha$; uncertainty from estimating the baseline decreases as the number $\ell$ of reference snapshots increases.

Define the baselined coefficient of variation by
\begin{equation}
\mathrm{CV}_\Delta
=
\frac{\sqrt{V_\Delta}}
{|\overline{y^\Delta}|}.
\end{equation}
As for the raw CV, the baseline affects its behavior through both the numerator and denominator. However, it does not change its ultimate critical behavior.

\begin{corollary}[CV under subtractive baselining]\label{cor:baseline-cv}
Under the sampling scheme above, let $\alpha_0$ be fixed as
$\alpha\to\alpha_{\mathrm c}^-$, and assume the hypotheses of
Theorem~\ref{thm:cv}. Assume additionally that $\Pr(\overline{y^\Delta}=0)=0$ for $\alpha<\alpha_{\mathrm c}$ sufficiently close to $\alpha_{\mathrm c}$. Then, the conclusions of
Theorem~\ref{thm:cv} remain valid with $\mathrm{CV}$ replaced by
$\mathrm{CV}_\Delta$. In particular, if
$0<\rho_{\mathrm c}<1$, then
\begin{equation}
\mathrm{CV}_\Delta
\longrightarrow
\sqrt{\frac{N(1-\rho_{\mathrm c})}
{(N-1)\rho_{\mathrm c}}}
\end{equation}
in probability. If $\rho_{\mathrm c}=1$, then
$\mathrm{CV}_\Delta\to0$, whereas if $\rho_{\mathrm c}=0$, then
$\mathrm{CV}_\Delta\to\infty$ in probability.
\end{corollary}

\begin{proof}
Lemma~\ref{lem:coordinate} gives
$\bx=\gamma\bv+\br$, where
$\br=O_{\mathbb P}(1)$ and
$|\gamma|\to\infty$ in probability.
Because $\alpha_0$ is fixed, the independently estimated baseline
$\bb$ is also $O_{\mathbb P}(1)$. Hence
\begin{equation}
\by^\Delta
=
\bx-\bb
=
\gamma\bv+O_{\mathbb P}(1),
\end{equation}
and therefore
\begin{align}
P\by^\Delta
&=
\gamma P\bv+O_{\mathbb P}(1),\\
\overline{y^\Delta}
&=
\gamma\bar v+O_{\mathbb P}(1).
\end{align}
These are the same leading relations used in the proof of
Theorem~\ref{thm:cv}, so the same limits of $\mathrm{CV}_\Delta$ follow.
\end{proof}

\begin{corollary}[Fixed baselines do not change the ratio limits]\label{cor:baseline-ratio}
Let $\bb\in\mathbb R^N$ be fixed and let $\by=\bx-\bb$. The statistics $g_1$, $g_2$, and $I_{\mathrm M}$ computed from $\by$ have the same leading limits as in Theorem~\ref{thm:ratios}.
\end{corollary}

\begin{proof}
Equation~\eqref{eq:centermean1} gives $P\bx=\gamma P\bv+O_{\mathbb P}(1)$. Because $P\bb$ is fixed, subtracting the baseline changes only the $O_{\mathbb P}(1)$ remainder and leaves the leading term $\gamma P\bv$ unchanged.
\end{proof}

Corollary~\ref{cor:baseline-ratio} explains why baseline subtraction can markedly improve variance-based indicators but not skewness, kurtosis, or Moran's $I$.

\section{Hopf bifurcations}\label{sec:hopf}

We have assumed that one real eigenvalue of the stability matrix approaches zero. At a Hopf bifurcation, the loss of stability instead occurs through a complex-conjugate pair whose real parts approach zero while their imaginary parts remain nonzero. The growing fluctuations are therefore associated with a two-dimensional real subspace rather than with one right eigenvector. For Hopf bifurcations, we derive the resulting covariance and state the consequences for the spatial variance, the CV, skewness, kurtosis, Moran's $I$, snapshot uncertainty, and baseline referencing.
Because the resulting algebra is longer than in the real-eigenvalue case while the conclusions are largely parallel, we show the Hopf theory in full in Appendix~\ref{app:hopf}.

Several conclusions from the simple real-eigenvalue case remain qualitatively the same: noise-excited critical slowing down produces a divergent spatial variance; subtractive baseline referencing does not change its leading critical behavior; skewness, kurtosis, and Moran's $I$ do not diverge. The main qualitative difference is that a Hopf bifurcation produces a two-dimensional rather than one-dimensional critical fluctuation. As a result, when the critical subspace has nonzero spatial mean, the CV, kurtosis, and Moran's $I$ generally retain non-degenerate limiting distributions. The relative uncertainty of the spatial variance can take values between $1$ and $\sqrt{2}$ rather than converging universally to $\sqrt{2}$.

\section{Conclusions}\label{sec:discussion}

We have developed a theory of spatial EWSs for stochastic dynamics on general networks. The theory clarifies both the information and the ambiguity carried by each statistic, including whether its direction of change is universal and whether its deterministic and stochastic contributions can be separated. One practical conclusion is that, among commonly used spatial statistics, the spatial variance is the one whose fluctuation contribution generically inherits the divergence caused by critical slowing down. 
Near a simple steady-state instability, the linearized fluctuations are dominated by a single critical mode. The left critical eigenvector determines how strongly noise excites this mode, whereas the right critical eigenvector determines its observed spatial pattern.
Because spatial centering retains only the non-uniform part of that pattern, the expected spatial variance acquires the leading term $\sigma_{\mathrm{eff}}^2(1-\rho)/[2\lambda_1(N-1)]$ and diverges when the limiting critical eigendirection is both noise-excited and non-uniform. Network heterogeneity makes a non-uniform critical eigenvector very likely although it does not guarantee non-uniformity. Except in special cases, the spatial CV cancels the growing amplitude shared by its numerator and denominator and converges to a finite value. Skewness, kurtosis, and Moran's $I$ likewise divide out the growing amplitude and approach network-dependent limits that provide no universal warning direction.

The same analysis provides theoretical support for the baseline referencing proposed in~\cite{Bandara2026arxiv}. Proposition~\ref{prop:decomp} shows that the expected spatial variance combines a deterministic contribution, $S(\alpha)$, arising from differences among equilibrium node states with a stochastic contribution, $\Phi(\alpha)$, arising from fluctuations about that equilibrium. The deterministic contribution vanishes for a uniform equilibrium but is generally present along a heterogeneous equilibrium branch. It can therefore cause the spatial variance to change even when the fluctuation contribution changes in the opposite direction or not at all. Subtractive baselining removes the equilibrium offset at the reference parameter while leaving the current fluctuation term intact (Proposition~\ref{prop:baseline}). This offset is precisely the confounding contribution that theories developed for homogeneous domains have not had to confront. Our results therefore support the baseline-referenced spatial variance as a natural default candidate for heterogeneous networks, consistent with~\cite{Bandara2026arxiv}, while not guaranteeing improved performance over every parameter range.

A principal limitation of this work is its reliance on linearization. As $\lambda_1\to0$ when $\alpha\to\alpha_{\mathrm c}^-$, stochastic excursions eventually become large enough to leave the local neighborhood of the equilibrium, so the Ornstein--Uhlenbeck approximation may lose accuracy before the deterministic bifurcation is reached. We have also assumed additive white noise, a scalar state at each node, and a static network, and we have analyzed the spatial statistics themselves rather than the performance of a particular detection algorithm. Finally, our analysis is confined to standard bifurcations; other types of transitions~\cite{kuznetsov2004elements, Ashwin2012PhilTransRSocA, Boettiger2013TheorEcol} require separate treatment.

\section*{Acknowledgments}

N.M. acknowledges financial support by the Japan Science and Technology Agency (JST) Moonshot R\&D (under Grant No.\ JPMJMS2021), the National Science Foundation (under grant no.\,2204936), and JSPS KAKENHI (under grant nos.\,JP 23H03414, 24K14840, and 24K03013).
During the preparation of this manuscript, the author used Claude Opus 5.0, University of Michigan (U-M) ChatGPT 5.5, and ChatGPT 5.5 Pro for language editing, code development, and mathematical analysis assistance. The author manually verified all code and manuscript text.

\appendix

\section{Proof of Lemma~\ref{lem:rank1}\label{app:rank1-proof}}

The stationary covariance has the integral representation \cite[Section~6.2]{Sarkka2019book}
\begin{equation}\label{eq:covariance-integral}
C
=
\int_0^\infty
e^{-Mt}BB^\top e^{-M^\top t}
\,\mathrm dt.
\end{equation}
Let $\mathcal P_1=\bv\bw^\top$ and
$\mathcal Q_1=I-\mathcal P_1$. Because $\bv$ and $\bw$ have finite limits, $\mathcal P_1$ remains bounded. Continuity of $M$ along the equilibrium branch and the spectral separation in~\eqref{eq:specgapass} imply that there exist constants $K,\eta_0>0$, independent of $\alpha$ sufficiently close to $\alpha_{\mathrm c}$, such that
\begin{equation}
\|e^{-Mt}\mathcal Q_1\|_2
\le
K e^{-\eta_0t}.
\end{equation}
We substitute
$e^{-Mt}=e^{-\lambda_1t}\mathcal P_1+e^{-Mt}\mathcal Q_1$
and its transpose into~\eqref{eq:covariance-integral}. The term obtained by taking the critical component from both $e^{-Mt}$ and $e^{-M^{\top} t}$ is
\begin{equation}
\int_0^\infty
e^{-2\lambda_1t}
\mathcal P_1BB^\top\mathcal P_1^\top
\,\mathrm dt
=
\frac{\sigma_{\mathrm{eff}}^2}{2\lambda_1}
\bv\bv^\top.
\end{equation}
Each remaining term has at least one factor
$e^{-Mt}\mathcal Q_1$, or its transpose, and therefore has a uniformly bounded integral as $\alpha\to\alpha_{\mathrm c}^-$. Therefore, we obtain~\eqref{eq:rank1}.

\section{Details of the numerical simulations}\label{app:sim}

This appendix describes the numerical simulations shown in Fig.~\ref{fig:ews} and the five networks used there. The simulation protocol follows the one in our previous study~\cite{Bandara2026arxiv}.

\subsection{Dynamical system}

We use the coupled double-well dynamics~\eqref{eq:doublewell} with $r_1=1$, $r_2=3$, $r_3=5$, and noise strength $\sigma=0.1$, so that an isolated noise-free node has stable equilibria at $x_i=1$ and $x_i=5$ and an unstable one at $x_i=3$. We drive the system toward a tipping point in two ways. First, we decrease the coupling strength $D$ with the stress fixed at $u=-5$. Second, we decrease $u$ with $D=0.05$ fixed. In both cases, all nodes start in the upper state, and the transition of interest is the collapse of that state.

\subsection{Numerical method}

For each network and each choice of control parameter (i.e., $D$ or $u$), we take $100$ equally spaced values of that parameter, which we call the simulation range. At each value, we integrate Eq.~\eqref{eq:doublewell} by the Euler--Maruyama scheme with time step $\Delta t$, initial condition $x_i(0)=5$ for every $i$, and final time $T=50$, chosen to allow the state to settle near the upper equilibrium with fluctuations caused by the dynamical noise. We set $\Delta t=0.01$ for the square lattice, Montreal, and jazz player networks, $\Delta t=0.005$ for the \textit{C.\ elegans} metabolic network, and $\Delta t = 0.001$ for the largest network, i.e., the US power grid. A smaller $\Delta t$ is necessary for the larger and denser networks because the coupling term makes the explicit Euler--Maruyama scheme unstable at a larger $\Delta t$.

We regard the state of the $i$th node at time $T$, denoted by $x_i(T)$, as an equilibrated sample in the presence of the dynamical noise. The single snapshot, $\{x_1(T),\ldots,x_N(T)\}$, is the data from which the five spatial EWSs are computed at each value of the control parameter.

Under descending $D$, the $100$ values of $D$ span $[0,1]$ for the square lattice, $[0.3,1]$ for Montreal and the US power grid, $[0.3,0.6]$ for the jazz player network, and $[0.3,0.8]$ for the \textit{C.\ elegans} metabolic network. Under descending $u$, they span $[-4,0]$ for all five networks. The simulation range differs across networks because it must contain the tipping point while leaving a sufficient number of values of the control parameter before it. We verified for each run that the noise-free system remains in the upper state at the far end of the range, that no state diverges, and that at least $30$ of the $100$ values of the control parameter precede the tipping point.

\subsection{Home range}

Equilibrated states recorded after a tipping event has occurred are not relevant to early warning. Therefore, we compute the EWSs only over the home range~\cite{maclaren2025applicability}, i.e., the contiguous sequence of values of the control parameter starting from the value farthest from the transition and ending just before the first value at which any node has left the upper state. We regard the $i$th node as remaining in the upper state as long as $x_i(T)>r_2=3$. The home range consists of $88$, $91$, $79$, $78$, and $82$ of the $100$ values of the control parameter under descending $D$, and $95$, $81$, $81$, $79$, and $78$ values under descending $u$, for the square lattice, Montreal, jazz player, \textit{C.\ elegans} metabolic, and US power grid networks, respectively. The dotted vertical line in each panel of Fig.~\ref{fig:ews} indicates the first value of the control parameter at which a node has left the upper state, marking the first value outside the home range. This control parameter value is approximately where the first tipping event occurs; EWSs aim to anticipate this event before it occurs.

\subsection{Networks}

We use five networks, i.e., one synthetic network and four empirical ones, which we summarize in Table~\ref{tab:networks}. We use the largest connected component of each network and treat it as undirected and unweighted, removing multiple edges and self-loops. The square lattice is a regular graph and hence satisfies the assumption of Theorem~\ref{thm:lattice}; the other four networks are degree-heterogeneous.

\begin{table}[t]
\centering
\caption{Networks used in Fig.~\ref{fig:ews}. $N$: number of nodes. $|E|$: number of edges.}
\label{tab:networks}
\begin{tabular}{@{}lrrp{0.52\textwidth}@{}}
\toprule
Network & $N$ & $|E|$ & Description\\
\midrule
Square lattice & 100 & 200 & Two-dimensional square lattice composed of $10\times10$ nodes with periodic boundary conditions. Each node is adjacent to its four nearest neighbors. It represents the spatially homogeneous domain used in much of the classical spatial EWS literature.\\
Montreal & 29 & 75 & Network of street gangs operating in Montreal, Quebec~\cite{descormiers2011alliances}. Each node is a street gang. An edge indicates an alliance or a rivalry between two gangs. The original data do not distinguish between the two types of relationship.\\
Jazz player & 198 & 2742 & Collaboration network of jazz musicians~\cite{gleiser2003community}. Each node is a musician. Two nodes are adjacent if the corresponding musicians performed together in the same band or ensemble.\\
\textit{C.\ elegans} metabolic & 453 & 2025 & Metabolic network of the nematode \textit{Caenorhabditis elegans}~\cite{jeong2000large}. Each node is a biochemical entity such as an enzyme, a metabolite, or a transient complex. An edge indicates that two entities participate in the same biochemical reaction.\\
US power grid & 4941 & 6594 & Infrastructure network representing the Western States power grid of the USA~\cite{watts1998collective}. Each node is a generator, a substation, or a transformer. An edge represents a transmission line connecting two of them.\\
\bottomrule
\end{tabular}
\end{table}

\section{Proof of Theorem~\ref{thm:varuncertainty}}\label{app:varuncertainty}



For a zero-mean Gaussian vector, Isserlis' formula gives
\begin{equation}
\E[z_i z_j z_k z_{\ell}]
=
C_{ij}C_{k\ell}
+
C_{ik}C_{j\ell}
+
C_{i\ell}C_{jk},
\end{equation}
which leads to
\begin{equation}
\operatorname{Cov}(z_i z_j,z_k z_{\ell})
=
C_{ik}C_{j\ell}
+
C_{i\ell}C_{jk}.
\end{equation}
Using
$\bz^\top P\bz=\sum_{i,j=1}^N P_{ij}z_i z_j$, we obtain
\begin{align}
\Var(\bz^\top P\bz)
&=
\sum_{i,j,k,\ell=1}^N
P_{ij}P_{k\ell}
\operatorname{Cov}(z_i z_j,z_k z_{\ell})\\
&=
\sum_{i,j,k,\ell=1}^N
P_{ij}P_{k\ell}
\left(
C_{ik}C_{j\ell}
+
C_{i\ell}C_{jk}
\right)\\
&=
2\Tr(PCPC),
\end{align}
where the last equality uses the symmetry of $P$ and $C$.
Because $V_{\mathrm{fluc}}=(\bz^\top P\bz)/(N-1)$, this gives
\eqref{eq:varVfluc}.

Recall that
we defined $s^2 = \sigma_{\mathrm{eff}}^2 / (2\lambda_1)$ in Lemma~\ref{lem:coordinate}.
Lemma~\ref{lem:rank1} gives
\begin{equation}
PC=s^2(P\bv)\bv^\top+O(1).
\label{eq:uncertainty-PC}
\end{equation}
Because $\bv^{\top} P \bv = 1-\rho\to1-\rho_{\mathrm c}>0$, we obtain
\begin{equation}
\Tr(PC)=s^2(1-\rho)+O(1)
\label{eq:TrPC-proof}
\end{equation}
and
\begin{equation}
\Tr(PCPC)=s^4(1-\rho)^2+O(s^2).
\label{eq:TrPCPC-s4}
\end{equation}
Therefore,
\begin{equation}
\frac{\sd[V_{\mathrm{fluc}}]}{\E[V_{\mathrm{fluc}}]}
=
\frac{\sqrt{2\Tr(PCPC)}}{\Tr(PC)}
\longrightarrow \sqrt{2}.
\label{eq:TrPCPC-TrPC-proof-convergence}
\end{equation}

It remains to show that the same limit holds for the full spatial variance $V$. Equation~\eqref{eq:V-realization-decomp} yields
\begin{equation}
V
=
S(\alpha)
+
\frac{2(P\bx^*)^\top P\bz}{N-1}
+
V_{\mathrm{fluc}}.
\label{eq:V-realization-decomp-2}
\end{equation}
The first term on the right-hand side of \eqref{eq:V-realization-decomp-2} is deterministic and bounded. The second term is linear in the zero-mean Gaussian vector $\bz$, so it has zero mean.
Using the fact that $\Var(\bu^\top\bz)=\bu^\top C\bu$ for a deterministic vector $\bu$, with $\bu=P^{\top}P\bx^*=P\bx^*$, we obtain
\begin{equation}
\Var \left(\frac{2(P\bx^*)^\top P\bz}{N-1}\right)
=
\frac{4}{(N-1)^2}
(P\bx^*)^\top P C P\bx^*,
\label{eq:var-of-second-term}
\end{equation}
where we used $P^\top=P$ and $P^2=P$.
By multiplying $P$ from the right to both sides of \eqref{eq:uncertainty-PC}, we obtain
\begin{equation}
PCP
=
s^2 (P\bv)(P\bv)^\top+O(1).
\end{equation}
Because $\bx^*$ and $\bv$ remain bounded near the bifurcation, we obtain
\begin{equation}
(P\bx^*)^\top P C P\bx^*
=
(P\bx^*)^\top P C P^2\bx^*
=
s^2\big[(P\bx^*)^\top P\bv\big]^2+O(1)
=
O(s^2).
\end{equation}
Therefore, \eqref{eq:var-of-second-term} is $O(s^2)$ as $\alpha \to \alpha_{\mathrm c}^-$.

Moreover, the covariance between the second term on the right-hand side of \eqref{eq:V-realization-decomp-2}
and $V_{\mathrm{fluc}}$ is zero because all third moments of a zero-mean Gaussian vector vanish.
Combining these results gives
\begin{equation}
\Var[V]
=
\Var[V_{\mathrm{fluc}}]+O(s^2).
\end{equation}
Similarly, Proposition~\ref{prop:decomp} and the boundedness of $S(\alpha)$ give
\begin{equation}
\E[V]
=
\E[V_{\mathrm{fluc}}]+O(1).
\end{equation}
Equations~\eqref{eq:TrPC-proof}, \eqref{eq:TrPCPC-s4}, and \eqref{eq:TrPCPC-TrPC-proof-convergence} imply $\E[V_{\mathrm{fluc}}]= O(s^2)$ and $\Var[V_{\mathrm{fluc}}] = O(s^4)$ when $\sigma_{\mathrm{eff,c}}^2>0$ and $\rho_{\mathrm c}<1$. Therefore, the bounded structural term and the linear cross term are asymptotically negligible, and we obtain
\begin{equation}
\frac{\sd[V]}{\E[V]}
\longrightarrow
\sqrt{2}.
\label{eq:CV-V}
\end{equation}


\section{Ratio baseline referencing}\label{app:ratio-baseline-referencing}

As introduced in Section~\ref{sec:baseline}, the ratio variant of baseline referencing divides each node state by its node-specific baseline before one computes the spatial EWS, so that one compares the ratios $x_i/b_i$ rather than the differences $x_i-b_i$. In this section, we analyze its behavior, focusing on $V$ for simplicity.

Let
\begin{equation}
D_b=\operatorname{diag}(b_1,\ldots,b_N).
\end{equation}
We assume that its entries are bounded away from zero, i.e., $|b_i|\ge b_{\min}>0$ for all $i$. We define
\begin{equation}
\by^{\mathrm{rel}}=D_b^{-1}\bx,
\end{equation}
so that $y_i^{\mathrm{rel}}=x_i/b_i$. We consider
\begin{equation}
V_{\mathrm{rel}}
=
\Var(\by^{\mathrm{rel}})
=
\frac{\|PD_b^{-1}\bx\|^2}{N-1}
\end{equation}
as a baseline referenced spatial variance.
Because $\bx=\bx^*(\alpha)+\bz$, we obtain
\begin{equation}
PD_b^{-1}\bx
=
PD_b^{-1}\bx^*(\alpha)
+
PD_b^{-1}\bz.
\label{eq:multiplicative1}
\end{equation}
Conditional on the fixed baseline, the second term on the right-hand side of~\eqref{eq:multiplicative1} is centered because $\E[\bz]=\boldsymbol0$. Therefore, the cross term has zero conditional expectation, and
\begin{equation}\label{eq:Vrel}
\E[V_{\mathrm{rel}}\mid\bb]
=
\Var\!\left(D_b^{-1}\bx^*(\alpha)\right)
+
\frac{1}{N-1}
\Tr\!\left(
PD_b^{-1}C(\alpha)D_b^{-1}
\right).
\end{equation}
To average the right-hand side of~\eqref{eq:Vrel} over a random baseline, we require the reciprocal factors $1/b_i$ and $1/(b_i b_j)$, including $1/b_i^2$ when $i=j$, to have finite absolute expectations. An unconstrained Gaussian approximation to the baseline generally does not satisfy this condition because it assigns positive density near zero; this is why we condition on a realized baseline and impose the lower bound on $|b_i|$.

The first term on the right-hand side of~\eqref{eq:Vrel} describes the node-to-node variation of the equilibrium values after each value has been divided by its own baseline.
This term approximately vanishes at $\alpha = \alpha_0$ because
$D_b^{-1}\bx^*(\alpha_0) \approx \bone$, and may remain small as $\alpha$ varies.
The second term is the contribution of the current stochastic fluctuations, with each component $z_i$ rescaled by $b_i$.

Using Lemma~\ref{lem:rank1}, the leading contribution of the second term in~\eqref{eq:Vrel} is
\begin{equation}\label{eq:Vrelleading}
\frac{\sigma_{\mathrm{eff}}^2}{2\lambda_1}
\Var\!\left(D_b^{-1}\bv\right).
\end{equation}
If $\sigma_{\mathrm{eff,c}}^2>0$ and $\Var(D_b^{-1}\bv_{\mathrm c})>0$, the leading rate is proportional to $\lambda_1^{-1}$, and its amplitude is determined by the spatial variation of the rescaled vector whose $i$th entry is $v_i/b_i$. Ratio baseline referencing can consequently increase or decrease the leading amplitude relative to the original spatial variance.

Ratio and subtractive baseline referencing have different benefits. Ratio referencing compares node states relative to their own baseline values and is invariant to node-specific changes of units. The invariance to node-specific units carries a cost. If some $b_i$ is close to zero, the corresponding equilibrium value and fluctuation are both strongly amplified. A small number of such nodes can then dominate $V_{\mathrm{rel}}$ and make it sensitive to baseline-estimation error. This is why the baseline values should be bounded away from zero.
In contrast, subtractive referencing measures absolute changes from those baseline values and is generally more stable when some baselines are close to zero.

If the baseline is estimated once and then held fixed during monitoring, and if $\sigma_{\mathrm{eff,c}}^2>0$ and
$\Var(D_b^{-1}\bv_{\mathrm c})>0$, then the leading random contribution to $V_{\mathrm{rel}}$ is proportional to $\gamma^2$.
Therefore, by the same comparison of quadratic, linear, and bounded terms as in the proof of Theorem~\ref{thm:varuncertainty}, the relative standard deviation of $V_{\mathrm{rel}}$ over current snapshots tends to $\sqrt2$.

\section{Spatial early warning signals near a Hopf bifurcation}\label{app:hopf}

This appendix shows the full analysis summarized in Section~\ref{sec:hopf}.

\subsection{Leading covariance and expected spatial variance}

Throughout Appendix~\ref{app:hopf}, we replace Assumption~\ref{ass:spec} by the following assumption.

\begin{assumption}[Simple Hopf pair]\label{ass:hopf}
We assume that the stable equilibrium branch has the finite limit $\bx_{\mathrm c}^*$ as a Hopf bifurcation is approached with $\alpha\to\alpha_{\mathrm c}^-$.
For $\alpha<\alpha_{\mathrm c}$ sufficiently close to $\alpha_{\mathrm c}$, suppose that the real matrix $M(\alpha)$ is diagonalizable over $\mathbb C$ and has a simple complex-conjugate eigenvalue pair
\begin{equation}\label{eq:hopfpair}
\lambda_+(\alpha)
=
\eta(\alpha)+\mathrm i\Omega(\alpha),
\qquad
\lambda_-(\alpha)
=
\eta(\alpha)-\mathrm i\Omega(\alpha),
\end{equation}
where
\begin{equation}
\eta(\alpha)>0,
\qquad
\eta(\alpha)\longrightarrow0,
\qquad
\Omega(\alpha)\longrightarrow\Omega_{\mathrm c}>0.
\end{equation}
All remaining eigenvalues satisfy
\begin{equation}
\operatorname{Re}\lambda_k(\alpha)\ge\delta>0.
\end{equation}
Let $\bv_{\mathrm H}$ and $\overline{\bv_{\mathrm H}}$ be the right eigenvectors associated with $\lambda_+$ and $\lambda_-$, respectively. Let $\bw_{\mathrm H}$ and $\overline{\bw_{\mathrm H}}$ be the corresponding left eigenvectors. We choose them as part of the biorthogonal families used in Section~\ref{sec:setup}, with
\begin{equation}
\|\bv_{\mathrm H}\|_2=1,
\qquad
\bw_{\mathrm H}^{\top}\bv_{\mathrm H}=1,
\qquad
\bw_{\mathrm H}^{\top}\overline{\bv_{\mathrm H}}=0.
\end{equation}
The vectors $\bv_{\mathrm H}$ and $\bw_{\mathrm H}$ have finite limits as $\alpha\to\alpha_{\mathrm c}^-$.
\end{assumption}

For a generic Hopf bifurcation, the transversality condition gives
\begin{equation}\label{eq:hopflinear}
\eta(\alpha)
=
\kappa_{\mathrm H}(\alpha_{\mathrm c}-\alpha)
\bigl(1+o(1)\bigr),
\qquad
\kappa_{\mathrm H}>0,
\end{equation}
under the parameter orientation used here~\cite{kuznetsov2004elements}.

Define the effective noise intensity of the Hopf pair by
\begin{equation}\label{eq:sigmaH}
\sigma_{\mathrm H}^2
=
\bw_{\mathrm H}^{\top}
BB^\top
\overline{\bw_{\mathrm H}}
=
\|B^\top\bw_{\mathrm H}\|_2^2.
\end{equation}
This quantity is real and non-negative. We write $\sigma_{\mathrm H}\ge0$ for its non-negative square root and define
\begin{equation}
\sigma_{\mathrm H,c}^2
=
\lim_{\alpha\to\alpha_{\mathrm c}^-}
\sigma_{\mathrm H}^2,
\end{equation}
which exists because $B$ is fixed and $\bw_{\mathrm H}$ has a finite limit. As in the steady-state case, the left eigenvector determines how strongly the noise excites the critical fluctuations.

We also write
\begin{align}\label{eq:pqH}
\bp &= \operatorname{Re}\bv_{\mathrm H},\\
\boldsymbol q &= \operatorname{Im}\bv_{\mathrm H}.
\end{align}
The vectors $\bp$ and $\boldsymbol q$ are linearly independent because $\bv_{\mathrm H}$ is associated with a nonreal eigenvalue of a real matrix.

\begin{lemma}[Leading covariance near a Hopf bifurcation]\label{lem:hopfcov}
Under Assumption~\ref{ass:hopf},
\begin{equation}\label{eq:hopfcov}
C
=
\frac{\sigma_{\mathrm H}^2}{2\eta}
\left(
\bv_{\mathrm H}\overline{\bv_{\mathrm H}}^{\top}
+
\overline{\bv_{\mathrm H}}\bv_{\mathrm H}^{\top}
\right)
+
O(1).
\end{equation}
Equivalently,
\begin{equation}\label{eq:hopfcovreal}
C
=
\frac{\sigma_{\mathrm H}^2}{\eta}
\left(
\bp\bp^\top
+
\boldsymbol q\boldsymbol q^\top
\right)
+
O(1),
\end{equation}
where the remainder is bounded in operator norm.
\end{lemma}

\begin{proof}
Although the covariance expansion~\eqref{eq:Cspectral} was introduced under Assumption~\ref{ass:spec}, its derivation uses only stability, diagonalizability, and the biorthogonal eigenvector expansion. It therefore applies under Assumption~\ref{ass:hopf} as well.

The terms in the covariance expansion with index pairs $(+,-)$ and $(-,+)$ have denominator
\begin{equation}
\lambda_++\lambda_-=2\eta,
\end{equation}
and
\begin{equation}
\beta_{+-}
=
\beta_{-+}
=
\bw_{\mathrm H}^{\top}BB^\top\overline{\bw_{\mathrm H}}
=
\sigma_{\mathrm H}^2.
\end{equation}
The terms with index pairs $(+,+)$ and $(-,-)$ remain bounded because their denominators approach $2\mathrm i\Omega_{\mathrm c}$ and $-2\mathrm i\Omega_{\mathrm c}$, respectively. The remaining terms involve at least one non-critical eigendirection. To control them, we use the integral representation of the covariance given by~\eqref{eq:covariance-integral}. In each of the two matrix-exponential factors, we separate the part acting on the Hopf eigenspace from the part acting on the subspace spanned by the non-critical eigenvectors. Because $M(\alpha)\to M_{\mathrm c}$ and the Hopf pair remains separated from the rest of the spectrum, every term in the covariance integral containing at least one non-critical part is $O(1)$. This proves~\eqref{eq:hopfcov}. Writing $\bv_{\mathrm H}=\bp+\mathrm i\boldsymbol q$ gives~\eqref{eq:hopfcovreal}.

\end{proof}

The quantity corresponding to $\rho$ in the steady-state theory is
\begin{equation}\label{eq:rhoH}
\rho_{\mathrm H}
=
\frac{|\bone^\top\bv_{\mathrm H}|^2}{N}.
\end{equation}
Because $\|\bv_{\mathrm H}\|_2=1$,
\begin{equation}\label{eq:centerH}
\|P\bp\|^2+\|P\boldsymbol q\|^2
=
1-\rho_{\mathrm H}.
\end{equation}
The quantity $\rho_{\mathrm H}$ is unchanged if $\bv_{\mathrm H}$ is multiplied by a complex scalar of unit modulus, which merely rotates the pair $\bp,\boldsymbol q$ within their two-dimensional real span.

For every $\alpha<\alpha_{\mathrm c}$ sufficiently close to the bifurcation, we obtain
\begin{equation}
0\le\rho_{\mathrm H}<1.
\label{eq:0<rhoH<1-Hopf}
\end{equation}
Indeed, $\rho_{\mathrm H} = 1$ would require $\bv_{\mathrm H}$ to be a complex multiple (including the case of a real multiple) of $\bone$, which is impossible for a real matrix with a nonreal eigenvalue because $M\bone$ is real.

Let
\begin{equation}
\bv_{\mathrm H,c}
=
\lim_{\alpha\to\alpha_{\mathrm c}^-}\bv_{\mathrm H}(\alpha)
\end{equation}
and define
\begin{equation}
\rho_{\mathrm H,c}
=
\frac{|\bone^\top\bv_{\mathrm H,c}|^2}{N}
=
\lim_{\alpha\to\alpha_{\mathrm c}^-}\rho_{\mathrm H}(\alpha).
\end{equation}
By continuity of the Jacobian along the equilibrium branch, $M(\alpha)\to M_{\mathrm c}\equiv M(\alpha_{\mathrm c})$. Passing to the limit in $M\bv_{\mathrm H}=\lambda_+\bv_{\mathrm H}$ gives
\begin{equation}
M_{\mathrm c}\bv_{\mathrm H,c}
=
\mathrm i\Omega_{\mathrm c}\bv_{\mathrm H,c}.
\end{equation}
Thus, $\rho_{\mathrm H,c}<1$ by the same argument that a complex multiple of $\bone$ cannot be an eigenvector of the real matrix $M_{\mathrm c}$ for the nonreal eigenvalue $\mathrm i\Omega_{\mathrm c}$. Applying the spatial-centering operator $P$ subtracts the mean across nodes from both $\bp_{\mathrm c}$ and $\boldsymbol q_{\mathrm c}$. These two centered vectors, $P\bp_{\mathrm c}$ and $P\boldsymbol q_{\mathrm c}$, cannot both vanish. Therefore, the growing Hopf fluctuation cannot be removed completely by subtracting the spatial mean. However, the two centered vectors may become linearly dependent, in which case the two-dimensional Hopf fluctuation is reduced to a single spatial direction after centering.

\begin{theorem}[Spatial variance near a Hopf bifurcation]\label{thm:hopfvar}
Under Assumption~\ref{ass:hopf},
\begin{equation}\label{eq:hopfEV}
\E[V]
=
S(\alpha)
+
\frac{\sigma_{\mathrm H}^2}{\eta}
\frac{1-\rho_{\mathrm H}}{N-1}
+
O(1).
\end{equation}
For a generic Hopf bifurcation satisfying~\eqref{eq:hopflinear} and $\sigma_{\mathrm H,c}^2>0$,
$\E[V]$ grows in proportion to $(\alpha_{\mathrm c}-\alpha)^{-1}$.
\end{theorem}

\begin{proof}
Using~\eqref{eq:hopfcovreal},
\begin{align}
\Tr(PC)
&=
\frac{\sigma_{\mathrm H}^2}{\eta}
\left[
\Tr(P\bp\bp^\top)
+
\Tr(P\boldsymbol q\boldsymbol q^\top)
\right]
+
O(1) \notag\\
&=
\frac{\sigma_{\mathrm H}^2}{\eta}
\left(
\|P\bp\|^2+\|P\boldsymbol q\|^2
\right)
+
O(1).
\end{align}
Equation~\eqref{eq:centerH}, followed by the exact decomposition~\eqref{eq:decomp}, gives~\eqref{eq:hopfEV}. Because $\sigma_{\mathrm H}^2\to\sigma_{\mathrm H,c}^2>0$ and $1-\rho_{\mathrm H}\to1-\rho_{\mathrm H,c}>0$, the stated divergence rate follows from~\eqref{eq:hopflinear}.
\end{proof}

Unlike a saddle-node bifurcation, a Hopf bifurcation does not make the equilibrium equation singular. At $\alpha=\alpha_{\mathrm c}$, the critical eigenvalues $\pm\mathrm i\Omega_{\mathrm c}$ are nonzero, and all remaining eigenvalues are bounded away from zero. Hence $M(\alpha_{\mathrm c})$ is invertible. Under the usual smoothness conditions for a Hopf bifurcation, the implicit function theorem therefore gives a smooth continuation of the equilibrium branch through the bifurcation. In particular, if $\varepsilon=\alpha_{\mathrm c}-\alpha$,
\begin{equation}\label{eq:hopfstruct}
S(\alpha)
=
S(\alpha_{\mathrm c})
+
O(\varepsilon).
\end{equation}
Thus, the structural term remains bounded and has no square-root singularity. When $\sigma_{\mathrm H,c}^2>0$, the divergence of the expected spatial variance comes from the fluctuation term in~\eqref{eq:hopfEV}.
If $\sigma_{\mathrm H,c}^2=0$, the rate at which $\sigma_{\mathrm H}^2$ approaches zero relative to $\eta$ determines whether the Hopf contribution diverges.

\subsection{Behavior of spatial EWSs near the Hopf bifurcation}

We write
\begin{equation}
\bp_{\mathrm c}
=
\operatorname{Re}\bv_{\mathrm H,c},
\qquad
\boldsymbol q_{\mathrm c}
=
\operatorname{Im}\bv_{\mathrm H,c}.
\end{equation}
Let us define
\begin{equation}\label{eq:hopfrandomvector}
\bu_{\mathrm H}
=
G_1\bp_{\mathrm c}
+
G_2\boldsymbol q_{\mathrm c},
\end{equation}
where $G_1$ and $G_2$ are independent standard normal random variables.

\begin{theorem}[Hopf snapshot limit]\label{thm:hopfsnapshot}
Assume that $\sigma_{\mathrm H,c}^2>0$. Then
\begin{equation}\label{eq:hopfsnapshot}
\frac{\sqrt{\eta}}{\sigma_{\mathrm H}}
\left(
\bx-\bx^*(\alpha)
\right)
\xrightarrow{\mathrm d}
\bu_{\mathrm H}.
\end{equation}
Consequently,
\begin{equation}\label{eq:hopfcenterlimit}
\frac{\sqrt{\eta}}{\sigma_{\mathrm H}}P\bx
\xrightarrow{\mathrm d}
P\bu_{\mathrm H},
\end{equation}
and
\begin{equation}\label{eq:hopfmeanlimit}
\frac{\sqrt{\eta}}{\sigma_{\mathrm H}}\bar x
\xrightarrow{\mathrm d}
\frac{1}{N}\bone^\top\bu_{\mathrm H}.
\end{equation}
\end{theorem}

\begin{proof}
The stationary fluctuation $\bz=\bx-\bx^*$ is Gaussian.
Since $\operatorname{Cov}(\bz)=C$, the covariance of the rescaled fluctuation, $\sqrt{\eta} \bz / \sigma_{\mathrm H}$, is
\begin{equation}
\operatorname{Cov}\left(
\frac{\sqrt{\eta}}{\sigma_{\mathrm H}}\bz
\right)
=
\frac{\eta}{\sigma_{\mathrm H}^2}C.
\end{equation}
Using~\eqref{eq:hopfcovreal}, we obtain
\begin{equation}
\frac{\eta}{\sigma_{\mathrm H}^2}C
=
\bp\bp^\top
+
\boldsymbol q\boldsymbol q^\top
+
O\left(\frac{\eta}{\sigma_{\mathrm H}^2}\right).
\end{equation}
Because $\sigma_{\mathrm H}^2$ tends to a positive limit,
$\eta\to0$, $\bp\to\bp_{\mathrm c}$, and
$\boldsymbol q\to\boldsymbol q_{\mathrm c}$, this covariance converges to
$\bp_{\mathrm c}\bp_{\mathrm c}^\top
+
\boldsymbol q_{\mathrm c}\boldsymbol q_{\mathrm c}^\top$.
This is the covariance of
$\bu_{\mathrm H}=G_1\bp_{\mathrm c}+G_2\boldsymbol q_{\mathrm c}$
because $G_1$ and $G_2$ are independent standard normal random variables.
Convergence of the centered Gaussian distributions therefore gives~\eqref{eq:hopfsnapshot}. Because $\bx^*(\alpha)$ is bounded and $\sqrt{\eta}/\sigma_{\mathrm H}\to0$, the rescaled equilibrium term $(\sqrt{\eta}/\sigma_{\mathrm H}) \bx^*(\alpha)$ converges to zero. Applying the linear maps $P$ and $N^{-1}\bone^\top$ to $\bx=(\bx-\bx^*)+\bx^*$ then gives~\eqref{eq:hopfcenterlimit} and~\eqref{eq:hopfmeanlimit}, respectively.

\end{proof}

Theorem~\ref{thm:hopfsnapshot} replaces the single-random-coordinate representation used for a steady-state bifurcation. A Hopf snapshot is governed by two independent Gaussian amplitudes.

We define the symmetric positive-semidefinite matrix
\begin{equation}\label{eq:hopfgram}
\mathcal A_{\mathrm H}
=
\begin{pmatrix}
\|P\bp_{\mathrm c}\|^2
&
(P\bp_{\mathrm c})^\top(P\boldsymbol q_{\mathrm c})
\\
(P\bp_{\mathrm c})^\top(P\boldsymbol q_{\mathrm c})
&
\|P\boldsymbol q_{\mathrm c}\|^2
\end{pmatrix}
\end{equation}
and
\begin{equation}\label{eq:hopfmeanvector}
\boldsymbol c_{\mathrm H}
=
\begin{pmatrix}
\bar p_{\mathrm c}\\
\bar q_{\mathrm c}
\end{pmatrix}.
\end{equation}
Then, we obtain
\begin{equation}
\Tr(\mathcal A_{\mathrm H})
=
1-\rho_{\mathrm H,c}
\end{equation}
and
\begin{equation}
N\|\boldsymbol c_{\mathrm H}\|^2
=
\rho_{\mathrm H,c}.
\end{equation}

Using the definition $V=\|P\bx\|^2/(N-1)$, we obtain
\begin{equation}
\frac{\eta}{\sigma_{\mathrm H}^2}V
=
\frac{1}{N-1}
\left\|
\frac{\sqrt{\eta}}{\sigma_{\mathrm H}}P\bx
\right\|^2.
\label{eq:V-Hopf-expression-1}
\end{equation}
The continuous mapping theorem applied to~\eqref{eq:hopfcenterlimit}, together with~\eqref{eq:V-Hopf-expression-1}, yields
\begin{equation}
\frac{\eta}{\sigma_{\mathrm H}^2}V
\xrightarrow{\mathrm d}
\frac{\|P\bu_{\mathrm H}\|^2}{N-1}.
\label{eq:PuH-converge-Hopf}
\end{equation}
Because
$P\bu_{\mathrm H}=G_1P\bp_{\mathrm c}+G_2P\boldsymbol q_{\mathrm c}$,
we obtain
\begin{equation}
\|P\bu_{\mathrm H}\|^2
=
\boldsymbol G^\top\mathcal A_{\mathrm H}\boldsymbol G,
\label{eq:PuH-GAG1-Hopf}
\end{equation}
where
\begin{equation}
\boldsymbol G
=
\begin{pmatrix}
G_1\\
G_2
\end{pmatrix}.
\end{equation}
Substitution of \eqref{eq:PuH-GAG1-Hopf} into \eqref{eq:PuH-converge-Hopf} yields
\begin{equation}\label{eq:hopfVdist}
\frac{\eta}{\sigma_{\mathrm H}^2}V
\xrightarrow{\mathrm d}
\frac{\boldsymbol G^\top\mathcal A_{\mathrm H}\boldsymbol G}{N-1}.
\end{equation}

For the following CV statements, assume that $\Pr(\bar x=0)=0$ for $\alpha<\alpha_{\mathrm c}$ sufficiently close to $\alpha_{\mathrm c}$, so that the CV is defined almost surely. The CV behaves differently from its steady-state counterpart. If
$\boldsymbol c_{\mathrm H}\ne\boldsymbol0$,
then~\eqref{eq:hopfmeanlimit} gives
\begin{equation}
\frac{\sqrt{\eta}}{\sigma_{\mathrm H}}\bar x
\xrightarrow{\mathrm d}
\overline{u_{\mathrm H}}
=
G_1\bar p_{\mathrm c}+G_2\bar q_{\mathrm c}
=
\boldsymbol c_{\mathrm H}^\top\boldsymbol G.
\label{eq:Hopf-ave-x}
\end{equation}
Because both $(\eta / \sigma_{\mathrm H}^2) V$ and $(\sqrt{\eta} / \sigma_{\mathrm H}) |\bar{x}|$ are continuous functions of the same scaled snapshot in Theorem~\ref{thm:hopfsnapshot}, these two quantities converge jointly:
\begin{equation}
\left(
\frac{\sqrt{\eta}}{\sigma_{\mathrm H}}\sqrt V,
\frac{\sqrt{\eta}}{\sigma_{\mathrm H}}|\bar x|
\right)
\xrightarrow{\mathrm d}
\left(
\sqrt{
\frac{\boldsymbol G^\top\mathcal A_{\mathrm H}\boldsymbol G}{N-1}
},
|\boldsymbol c_{\mathrm H}^\top\boldsymbol G|
\right).
\end{equation}
Because the limit of $(\sqrt{\eta}  / \sigma_{\mathrm H}) |\bar{x}|$ is nonzero with probability one, the continuous mapping theorem gives
\begin{equation}\label{eq:hopfcvlimit}
\mathrm{CV}
=
\frac{
(\sqrt{\eta}/\sigma_{\mathrm H})\sqrt V
}{
(\sqrt{\eta}/\sigma_{\mathrm H})|\bar x|
}
\xrightarrow{\mathrm d}
\frac{
\sqrt{\boldsymbol G^\top\mathcal A_{\mathrm H}\boldsymbol G/(N-1)}
}{
|\boldsymbol c_{\mathrm H}^\top\boldsymbol G|
}.
\end{equation}
Because
$\boldsymbol c_{\mathrm H}^\top\boldsymbol G
=
G_1 \bar p_{\mathrm c} + G_2 \bar q_{\mathrm c}$
is Gaussian with variance
$\bar p_{\mathrm c}^{\,2}+\bar q_{\mathrm c}^{\,2}>0$,
the limiting denominator is nonzero with probability one.
The right-hand side of~\eqref{eq:hopfcvlimit} is generally a non-degenerate random variable. The numerator and denominator depend differently on the two Gaussian random variables. As in Section~\ref{sec:uncertainty}, the Gaussian denominator can be arbitrarily close to zero, so moments of the limiting CV are generally not finite.

If $\boldsymbol c_{\mathrm H}=\boldsymbol0$, or equivalently $\rho_{\mathrm H,c}=0$, then
$(\sqrt{\eta}/\sigma_{\mathrm H})|\bar x|$
converges to zero in probability. At the same time,
$(\sqrt{\eta}/\sigma_{\mathrm H})\sqrt V$
converges in distribution to
$\sqrt{\boldsymbol G^\top\mathcal A_{\mathrm H}\boldsymbol G/(N-1)}$.
This limiting random variable is strictly positive with probability one because $\mathcal A_{\mathrm H}$ is positive semidefinite and
$\Tr(\mathcal A_{\mathrm H})=1-\rho_{\mathrm H,c}=1$.
Combining the positive limiting numerator with the denominator's convergence to zero gives
$\mathrm{CV}\longrightarrow\infty$
in probability.

The remaining spatial EWSs are normalized statistics. As in the real-eigenvalue case discussed in the main text, this normalization divides out the growing critical amplitude. The difference at a Hopf bifurcation is that what survives is not a single fixed eigenvector but a random linear combination of the two critical directions, so the limits below are random variables rather than deterministic constants.

\begin{theorem}[Skewness, kurtosis, and Moran's $I$ near a Hopf bifurcation]\label{thm:hopfratios}
Under the hypotheses of Theorem~\ref{thm:hopfsnapshot}, let $\Theta$ be uniformly distributed on $[0,2\pi)$ and define
\begin{equation}\label{eq:hopfphasevector}
\bu_\Theta
=
(\cos\Theta)\bp_{\mathrm c}
+
(\sin\Theta)\boldsymbol q_{\mathrm c}.
\end{equation}
Here, as in Section~\ref{sec:setup}, $\mu_k(\bu_\Theta)$ denotes the $k$th spatial central moment across the $N$ components of $\bu_\Theta$, i.e.,
\begin{equation}
\mu_k(\bu_\Theta)
=
\frac{1}{N}
\sum_{i=1}^N(P\bu_\Theta)_i^k.
\end{equation}
For each realization of $\Theta$, this quantity is a scalar.
Then,
\begin{equation}\label{eq:hopfg1}
g_1
\xrightarrow{\mathrm d}
\frac{\mu_3(\bu_\Theta)}
{\mu_2(\bu_\Theta)^{3/2}},
\end{equation}
\begin{equation}\label{eq:hopfg2}
g_2
\xrightarrow{\mathrm d}
\frac{\mu_4(\bu_\Theta)}
{\mu_2(\bu_\Theta)^2},
\end{equation}
and
\begin{equation}\label{eq:hopfmoran}
I_{\mathrm M}
\xrightarrow{\mathrm d}
\frac{N}{W}
\frac{
(P\bu_\Theta)^\top A(P\bu_\Theta)
}{
(P\bu_\Theta)^\top(P\bu_\Theta)
}.
\end{equation}
The quantities in the denominators are nonzero with probability one. None of these three EWSs diverges.
\end{theorem}

\begin{proof}
We write
\begin{equation}
R
=
\sqrt{G_1^2+G_2^2}.
\end{equation}
Then, $R>0$ with probability one, and we can write
\begin{equation}
\begin{cases}
G_1 &= R\cos\Theta,\\
G_2 &= R\sin\Theta,
\end{cases}
\end{equation}
where $\Theta$ is uniformly distributed on $[0,2\pi)$. Therefore,
\begin{equation}
\bu_{\mathrm H}
=
R\bu_\Theta.
\end{equation}

Using~\eqref{eq:hopfcenterlimit}, we obtain
\begin{equation}
\frac{\sqrt{\eta}}{\sigma_{\mathrm H}}P\bx
\xrightarrow{\mathrm d}
P\bu_{\mathrm H}
=
R P\bu_\Theta.
\end{equation}
Skewness, kurtosis, and Moran's $I$ are unchanged when the centered snapshot is multiplied by the positive scalar $R$. Therefore, the common factor $R$ cancels from all three statistics. Applying the continuous mapping theorem gives~\eqref{eq:hopfg1}, \eqref{eq:hopfg2}, and \eqref{eq:hopfmoran}.

Finally, $P\bp_{\mathrm c}$ and $P\boldsymbol q_{\mathrm c}$ cannot both vanish, because that would imply $\rho_{\mathrm H,c}=1$, which is impossible for a Hopf bifurcation (see \eqref{eq:0<rhoH<1-Hopf}). Therefore,
\begin{equation}
P\bu_\Theta
=
(\cos\Theta)P\bp_{\mathrm c}
+
(\sin\Theta)P\boldsymbol q_{\mathrm c}
\end{equation}
can vanish, if at all, only for isolated values of $\Theta$. Because $\Theta$ has a continuous uniform distribution, these values of $\Theta$ occur with probability zero. Moreover,
\begin{equation}
\mu_2(\bu_\Theta)
=
\frac{1}{N}\|P\bu_\Theta\|^2.
\end{equation}
Therefore, $\mu_2(\bu_\Theta)>0$ with probability one, guaranteeing that the denominators in~\eqref{eq:hopfg1}, \eqref{eq:hopfg2}, and \eqref{eq:hopfmoran} are nonzero with probability one.

\end{proof}

A snapshot retains a random direction within the two-dimensional real critical subspace spanned by $\bp_{\mathrm c}$ and $\boldsymbol q_{\mathrm c}$. For this reason, Theorem~\ref{thm:hopfratios} shows that $g_1$, $g_2$, and Moran's $I$
generally remain random even arbitrarily close to the bifurcation. This result differs from the rank-one result in Theorem~\ref{thm:ratios}, especially for kurtosis and Moran's $I$, which converge to deterministic constants.

The transformation $\Theta\mapsto\Theta+\pi$ changes $\bu_\Theta$ to $-\bu_\Theta$. Therefore, the limiting skewness distribution is symmetric about zero. Because the sample skewness is bounded for fixed $N$, convergence in distribution also gives convergence of its expectation, and hence
$\E[g_1]\longrightarrow0$.

\subsection{Homogeneous networks}

The symmetric homogeneous scalar cases treated in Theorem~\ref{thm:lattice} and Corollary~\ref{cor:inputcoupling} have stability matrices of the form $aI-\omega A$ and hence have only real eigenvalues. They therefore cannot undergo a Hopf bifurcation under our hypotheses. Within the scalar-state framework of~\eqref{eq:sde}, a Hopf bifurcation requires a non-symmetric stability matrix.

\subsection{Snapshot uncertainty}

The rank-two structure also changes the snapshot-to-snapshot uncertainty of the spatial variance.

\begin{theorem}[Relative uncertainty at a Hopf bifurcation]\label{thm:hopfuncertainty}
Under the hypotheses of Theorem~\ref{thm:hopfsnapshot},
\begin{equation}\label{eq:hopfrelativeuncertainty}
\frac{\sd[V]}{\E[V]}
\longrightarrow
\frac{
\sqrt{2\Tr(\mathcal A_{\mathrm H}^2)}
}{
\Tr(\mathcal A_{\mathrm H})
}.
\end{equation}
If the two eigenvalues of $\mathcal A_{\mathrm H}$ are denoted by $\nu_1$ and $\nu_2$, then the limit is
\begin{equation}
\frac{\sqrt{2(\nu_1^2+\nu_2^2)}}{\nu_1+\nu_2}.
\end{equation}
This limiting value lies between $1$ and $\sqrt2$.
\end{theorem}

\begin{proof}
Let
\begin{equation}
r_{\mathrm H}(\alpha)
=
\frac{\sigma_{\mathrm H}^2}{\eta}.
\end{equation}
Under the hypotheses of Theorem~\ref{thm:hopfsnapshot},
$r_{\mathrm H}=\sigma_{\mathrm H}^2/\eta\to\infty$.
Equation~\eqref{eq:hopfcovreal} gives
\begin{equation}
C
=
r_{\mathrm H}
\left(
\bp\bp^\top+\boldsymbol q\boldsymbol q^\top
\right)
+
O(1).
\end{equation}
Therefore, we obtain
\begin{align}
\Tr(PC)
&=
r_{\mathrm H}
\left(
\|P\bp\|^2+\|P\boldsymbol q\|^2
\right)
+
O(1) \notag\\
&=
r_{\mathrm H}
\left[
\Tr(\mathcal A_{\mathrm H})+o(1)
\right]
+
O(1),
\label{eq:TrPC-Hopf-snapshot}
\end{align}
where the second equality follows from
$\bp\to\bp_{\mathrm c}$ and
$\boldsymbol q\to\boldsymbol q_{\mathrm c}$.

For computing $\Tr(PCPC)$, we write
$Q_{\mathrm H}=\bp\bp^\top+\boldsymbol q\boldsymbol q^\top$.
Expanding
$C=r_{\mathrm H}Q_{\mathrm H}+O(1)$ gives
\begin{equation}
\Tr(PCPC)
=
r_{\mathrm H}^2
\Tr(PQ_{\mathrm H}PQ_{\mathrm H})
+
O(r_{\mathrm H}).
\end{equation}
Moreover, it follows that
\begin{align}
\Tr(PQ_{\mathrm H}PQ_{\mathrm H})
&=
\|P\bp\|^4
+
2\left[(P\bp)^\top P\boldsymbol q\right]^2
+
\|P\boldsymbol q\|^4 \notag\\
&=
\Tr(\mathcal A_{\mathrm H}^2)+o(1).
\end{align}
Therefore, we obtain
\begin{equation}
\Tr(PCPC)
=
r_{\mathrm H}^2
\left[
\Tr(\mathcal A_{\mathrm H}^2)+o(1)
\right]
+
O(r_{\mathrm H}).
\label{eq:TrPCPC-Hopf-snapshot}
\end{equation}

By substituting
\eqref{eq:TrPC-Hopf-snapshot} and \eqref{eq:TrPCPC-Hopf-snapshot}
into \eqref{eq:meanVfluc} and~\eqref{eq:varVfluc}, respectively, we obtain
\begin{equation}
\frac{\sd[V_{\mathrm{fluc}}]}
{\E[V_{\mathrm{fluc}}]}
\longrightarrow
\frac{
\sqrt{2\Tr(\mathcal A_{\mathrm H}^2)}
}{
\Tr(\mathcal A_{\mathrm H})
}.
\end{equation}

It remains to pass from $V_{\mathrm{fluc}}$ to the full spatial variance. The realization-level decomposition~\eqref{eq:V-realization-decomp} consists of the bounded deterministic term $S(\alpha)$, a term linear in $\bz$, and $V_{\mathrm{fluc}}$. From~\eqref{eq:hopfcovreal}, the variance of the linear term is $O(r_{\mathrm H})$, whereas $\Var[V_{\mathrm{fluc}}]=O(r_{\mathrm H}^2)$. Their covariance is zero because all third moments of the centered Gaussian vector $\bz$ vanish. Therefore, we conclude
\begin{equation}
\Var[V]
=
\Var[V_{\mathrm{fluc}}]
+
O(r_{\mathrm H})
\end{equation}
and
\begin{equation}
\E[V]
=
\E[V_{\mathrm{fluc}}]
+
O(1).
\end{equation}
Thus, the full spatial variance has the same limiting relative standard deviation. Finally, because $\nu_1,\nu_2\ge0$ and are not both zero,
\begin{equation}
\frac{(\nu_1+\nu_2)^2}{2}
\le
\nu_1^2+\nu_2^2
\le
(\nu_1+\nu_2)^2.
\end{equation}
These inequalities give the lower and upper bounds $1$ and $\sqrt2$, respectively.
\end{proof}

Theorem~\ref{thm:hopfuncertainty} implies that the relative uncertainty need not approach $\sqrt{2}$ at a Hopf bifurcation, in contrast to the case of the rank-one bifurcations.
The spatial variance can be more precise (i.e., $\sd[V] / \E[V] < \sqrt{2}$)
when two centered directions,
$\bp_{\mathrm c}$ and $\boldsymbol q_{\mathrm c}$,
contribute comparably.

\subsection{Baseline referencing}

The exact subtractive-baseline decomposition in Proposition~\ref{prop:baseline} does not depend on the bifurcation type. If the reference parameter $\alpha_0$ is fixed away from the Hopf bifurcation, then
\begin{equation}
\E[V_\Delta]
=
\Var(\bde(\alpha))
+
\Phi(\alpha)
+
\frac{\Phi(\alpha_0)}{\ell},
\end{equation}
and the fluctuation term for the Hopf bifurcation satisfies
\begin{equation}
\Phi(\alpha)
=
\frac{\sigma_{\mathrm H}^2}{\eta}
\frac{1-\rho_{\mathrm H}}{N-1}
+
O(1).
\end{equation}
Thus, the divergent fluctuation term $\Phi(\alpha)$ is unchanged by subtractive baseline referencing. Compared with the raw decomposition
$\E[V]=S(\alpha)+\Phi(\alpha)$,
baselining replaces only the bounded structural term
$S(\alpha)=\Var(\bx^*(\alpha))$
by
$\Var(\bx^*(\alpha)-\bx^*(\alpha_0))$
and adds the constant baseline-estimation term
$\Phi(\alpha_0)/\ell$.
Moreover, the estimated subtractive baseline is $O_{\mathbb P}(1)$ when $\alpha_0$ is fixed, whereas the Hopf fluctuation grows as $\eta^{-1/2}$. Therefore, subtracting the baseline does not change the CV limit in~\eqref{eq:hopfcvlimit}, the divergence of the CV when $\rho_{\mathrm H,c}=0$, or the limiting distributions of skewness, kurtosis, and Moran's $I$ in Theorem~\ref{thm:hopfratios}.

For a fixed ratio baseline with diagonal matrix $D_b$ whose entries are bounded away from zero, the corresponding Hopf formulas are obtained by replacing
$\bp_{\mathrm c}$ and $\boldsymbol q_{\mathrm c}$
by $D_b^{-1}\bp_{\mathrm c}$ and $D_b^{-1}\boldsymbol q_{\mathrm c}$,
respectively, before spatial centering and then recomputing the corresponding Gram matrix and spatial-mean vector. In particular, the leading fluctuation contribution to the expected ratio-referenced variance is
\begin{equation}
\frac{\sigma_{\mathrm H}^2}{\eta(N-1)}
\left[
\|PD_b^{-1}\bp_{\mathrm c}\|^2
+
\|PD_b^{-1}\boldsymbol q_{\mathrm c}\|^2
+
o(1)
\right].
\end{equation}

\bibliographystyle{unsrt}
\bibliography{refs-theory}

\end{document}